\documentclass[a4paper,USenglish,cleveref,autoref,thm-restate]{lipics-v2021}

\usepackage{booktabs}

\usepackage{etoolbox}
\newtoggle{full}
\toggletrue{full}

\DeclareMathOperator{\connector}{c}
\DeclareMathOperator{\decomposition}{DT}
\DeclareMathOperator{\size}{size}
\newcommand{\prefix}[0]{\mathit{prefixsum}}
\newcommand{\partitionset}[0]{\mathit{PS}}
\DeclareMathOperator{\dist}{dist}
\DeclareMathOperator{\weight}{w}
\DeclareMathOperator{\vis}{vis}

\DeclareMathOperator{\neighborhood}{N}

\DeclareMathOperator{\portal}{portal}

\DeclareMathOperator{\argmin}{argmin}
\DeclareMathOperator{\degree}{deg}
\DeclareMathOperator{\subtree}{subtree}
\DeclareMathOperator{\proj}{proj}

\newcommand{\projy}[1]{\proj_y(#1)}
\newcommand{\projz}[1]{\proj_z(#1)}
\DeclareMathOperator{\neighbor}{n}

\newcommand{\neighbory}[1]{\neighbor_y(#1)}
\newcommand{\neighborz}[1]{\neighbor_z(#1)}
\newcommand{\region}[0]{Y}

\usepackage{todonotes}

\newcommand{\Geqt}{\ensuremath{G_{\Delta}}}

\iftoggle{full}{
\title{Polylogarithmic Time Algorithms for Shortest Path Forests in Programmable Matter}
\titlerunning{Polylogarithmic Time Algorithms for Shortest Path Forests in Programmable Matter}
}{
\title{Polylogarithmic Time Algorithms for Shortest Path Forests in Programmable Matter}
}

\author{Andreas Padalkin}{Paderborn University, Germany}{andreas.padalkin@upb.de}{https://orcid.org/0000-0002-4601-9597}{}
\author{Christian Scheideler}{Paderborn University, Germany}{scheideler@upb.de}{https://orcid.org/0000-0002-5278-528X}{}

\authorrunning{A. Padalkin, and C. Scheideler}

\Copyright{Andreas Padalkin, and Christian Scheideler}



\category{} 

\relatedversion{} 

\funding{This work was supported by the DFG Project SCHE 1592/10-1.}

\acknowledgements{The research of this paper was initialized at the Dagstuhl Seminar 23091 ``Algorithmic Foundations of Programmable Matter''. We thank Shantanu Das, Yuval Emek, Maria Kokkou, Irina Kostitsyna, Tom Peters, and Andrea Richa for helpful discussions.}

\pdfoutput=1
\hideLIPIcs
\nolinenumbers 

\EventEditors{John Q. Open and Joan R. Access}
\EventNoEds{2}
\EventLongTitle{42nd Conference on Very Important Topics (CVIT 2016)}
\EventShortTitle{CVIT 2016}
\EventAcronym{CVIT}
\EventYear{2016}
\EventDate{December 24--27, 2016}
\EventLocation{Little Whinging, United Kingdom}
\EventLogo{}
\SeriesVolume{42}
\ArticleNo{23}

\begin{document}

\maketitle

\begin{abstract}
    In this paper, we study the computation of shortest paths within the \emph{geometric amoebot model}, a commonly used model for programmable matter.
    Shortest paths are essential for various tasks and therefore have been heavily investigated in many different contexts.
    For example, in the programmable matter context, which is the focus of this paper, Kostitsyna et al.\ have utilized shortest path trees to transform one amoebot structure into another [DISC, 2023].
    We consider the \emph{reconfigurable circuit extension} of the model where this amoebot structure is able to interconnect amoebots by so-called circuits.
    These circuits permit the instantaneous transmission of simple signals between connected amoebots.
    
    We propose two distributed algorithms for the \emph{shortest path forest problem} where, given a set of $k$ sources and a set of $\ell$ destinations, the amoebot structure has to compute a forest that connects each destination to its closest source on a shortest path.
    For hole-free structures, the first algorithm constructs a shortest path tree for a single source within $O(\log \ell)$ rounds, and the second algorithm a shortest path forest for an arbitrary number of sources within $O(\log n \log^2 k)$ rounds.
    The former algorithm also provides an $O(1)$ rounds solution for the \emph{single pair shortest path problem} (SPSP) and an $O(\log n)$ rounds solution for the \emph{single source shortest path problem} (SSSP) since these problems are special cases of the considered problem.
\end{abstract}

\section{Introduction}

Programmable matter is matter that has the ability to change its physical properties in a programmable fashion \cite{DBLP:journals/ijhsc/ToffoliM93}.
Many exciting applications have already been envisioned for programmable matter such as self-healing structures \cite{an2021self} and minimal invasive surgery \cite{montemagno1999constructing}, and shape-changing robots have already been prominent examples of the potentials of programmable matter in many blockbuster movies.

In the \emph{amoebot model}, the matter consists of simple particles called \emph{amoebots}.
In the geometric variant of the model, the amoebots form a connected amoebot structure on the infinite triangular grid, on which they move by \emph{expansions} and \emph{contractions}.
However, since information can only travel amoebot by amoebot, many problems come with a natural lower bound of $\Omega(\mathit{diam})$ where $\mathit{diam}$ is the diameter of the structure.

For that reason, we consider the \emph{reconfigurable circuit extension} to the amoebot model where the amoebot structure is able to interconnect amoebots by so-called \emph{circuits}.
These circuits permit the instantaneous transmission of simple signals between connected amoebots.
The extension allows polylogarithmic solutions for various fundamental problems, e.g., leader election \cite{DBLP:journals/jcb/FeldmannPSD22}, and spanning tree construction \cite{DBLP:conf/dna/PadalkinSW22}.

In this paper, we consider the \emph{shortest path forest problem} where, given a set of sources and a set of destinations, the amoebots have to find a shortest path from each destination to the closest source.
Shortest paths are a fundamental problem in both centralized and distributed systems and have also a number of important applications in the amoebot model.

For example, consider shape formation.
Many algorithms for the amoebot model utilize a canonical shape, e.g., a line, as an intermediate structure \cite{DBLP:conf/spaa/DerakhshandehGR16,DBLP:journals/dc/LunaFSVY20}.
However, this is rather inefficient if the target structure is already close to the initial structure.
For such cases, Kostitsyna et al.\ proposed an algorithm that utilizes shortest paths to move amoebots through the structure to their target positions \cite{DBLP:conf/wdag/KostitsynaPS23}.
Another application for shortest paths is energy distribution \cite{DBLP:conf/icdcn/DaymudeRW21,DBLP:journals/corr/abs-2309-04898}.
The amoebots require energy to perform their movements that can be provided by other amoebots, e.g., amoebots located at external energy sources.
In order to minimize energy loss, it is more efficient to transfer the energy via a shortest path.

\subsection{Geometric Amoebot Model}
\label{sec:model:amoebot}

The \emph{(geometric) amoebot model} was proposed by Derakhshandeh et al.\ \cite{DBLP:conf/spaa/DerakhshandehDGRSS14}.
We will explain the model to a level of detail that is sufficient to understand the results of this paper.
For all other (unused) features of the model, e.g., movements, we refer to \cite{DBLP:journals/dc/DaymudeRS23,DBLP:conf/spaa/DerakhshandehDGRSS14}.

The model places a set of $n$ anonymous finite state machines called \emph{amoebots} on some graph $G = (V, E)$.
Each amoebot occupies one node and every node is occupied by at most one amoebot.
Let the \emph{amoebot structure} $X \subseteq V$ be the set of nodes occupied by the amoebots.
By abuse of notation, we identify amoebots with their nodes.
We assume that $G_X = (X, E_X)$ is connected, where $G_X = G|_X$ is the graph induced by $X$.
We call two amoebots that occupy adjacent nodes in $G$ \emph{neighbors} and denote the neighborhood of an amoebot $u \in X$ by $\neighborhood(u)$.

In the geometric variant of the model, $G$ is the infinite regular triangular grid graph $\Geqt = (V_\Delta, E_\Delta)$ (see \Cref{fig:model:classic}).
We assume that $X$ has no holes, i.e, $G_{V_\Delta \setminus X}$ is connected, where $G_{V_\Delta \setminus X} = \Geqt|_{V_\Delta \setminus X}$ is the graph induced by $V_\Delta \setminus X$.
Furthermore, each amoebot has a compass orientation (it defines one of its incident edges in $\Geqt$ as the northern direction) and a chirality.
We assume that all amoebots have the same compass orientation and chirality.
This is reasonable since Feldmann et al.\ \cite{DBLP:journals/jcb/FeldmannPSD22} showed that all amoebots are able to quickly come to an agreement within the considered extension (see \cref{sec:coordination}).
Our main focus will be the geometric variant but we will also propose some primitives for tree structures.

\begin{figure}[tbp]
    \begin{minipage}[t]{.48\linewidth}
        \centering
        \includegraphics{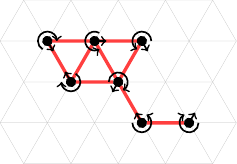}
        \subcaption{Amoebot structure.}
        \label{fig:model:classic}
    \end{minipage}
    \hfill
    \begin{minipage}[t]{.48\linewidth}
        \centering
        \includegraphics{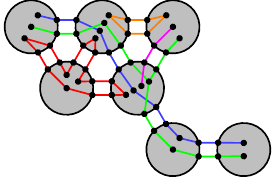}
        \subcaption{Reconfigurable circuit extension.}
        \label{fig:model:circuits}
    \end{minipage}
    \caption{
        Amoebot Model.
        The left figure shows an amoebot structure.
        The nodes indicate $X$.
        The red edges indicate $E_X$.
        The arrows indicate the compass orientation and chirality of each amoebot.
        The right figure shows an amoebot structure with $c = 2$ external links between adjacent amoebots.
        The amoebots are shown in gray.
        The nodes on the boundary are the pins.
        The nodes within the amoebots indicate the partition sets.
        An edge between a partition set $\partitionset$ and a pin $p$ implies $p \in \partitionset$.
        Each color indicates another circuit.
        The figures were taken and modified from \cite{DBLP:journals/jcb/FeldmannPSD22}.
    }
    \label{fig:model}
\end{figure}

\subsection{Reconfigurable Circuit Extension}
\label{sec:model:circuits}

In the \emph{reconfigurable circuit extension} \cite{DBLP:journals/jcb/FeldmannPSD22},
each edge between two neighboring amoebots $u$ and $v$ is replaced by $c$ edges called \emph{external links} with endpoints called \emph{pins}, for some constant $c \ge 1$ that is the same for all amoebots.
For each of these links, one pin is owned by $u$ while the other pin is owned by $v$.
In this paper, we assume that neighboring amoebots have a common labeling of their incident external links.

Each amoebot $u$ \emph{partitions} its \emph{pin set} $\partitionset(u)$ into a collection $\mathcal C(u)$ of pairwise disjoint subsets such that the union equals the pin set, i.e., $\partitionset(u) = \bigcup_{C \in \mathcal C(u)} C$.
We call $\mathcal C(u)$ the \emph{pin configuration} of $u$ and $C \in \mathcal C(u)$ a \emph{partition set} of $u$.
Let $\mathcal C = \bigcup_{u \in S} \mathcal C(u)$ be the collection of all partition sets in the system.
Two partition sets are \emph{connected} iff there is at least one external link between those sets.
Let $L$ be the set of all connections between the partition sets in the system.
Then, we call $H=(\mathcal C,L)$ the \emph{pin configuration} of the system and any connected component $C$ of $H$ a \emph{circuit} (see Figure~\ref{fig:model:circuits}).
Note that if each partition set of $\mathcal C$ is a \emph{singleton}, i.e., a set with exactly one element, then every circuit of $H$ just connects two neighboring amoebots.
An amoebot is part of a circuit iff the circuit contains at least one of its partition sets.
A priori, an amoebot $u$ may not know whether two of its partition sets belong to the same circuit or not since initially it only knows $\mathcal C(u)$.

Each amoebot $u$ can send a primitive signal (a \emph{beep}) via any of its partition sets $C \in \mathcal C(u)$ that is received by all partition sets of the circuit containing $C$ at the beginning of the next round.
The amoebots are able to distinguish between beeps arriving at different partition sets.
More specifically, an amoebot receives a beep at partition set $C$ if at least one amoebot sends a beep on the circuit belonging to $C$, but the amoebots neither know the origin of the signal nor the number of origins.


We assume the fully synchronous activation model, i.e., the time is divided into synchronous rounds, and every amoebot is active in each round.
On activation, each amoebot may update its state, reconfigure its pin configuration, and activate an arbitrary number of its partition sets.
The beeps are propagated on the updated pin configurations.
The time complexity of an algorithm is measured by the number of synchronized rounds required by it.


\subsection{Problem Statement and Our Contribution}

Let $S, D \subseteq X$ be two non-empty subsets.
We call each amoebot in $S$ a \emph{source}, and each amoebot in $D$ a \emph{destination}.
A \emph{$(S,D)$-shortest path forest} is a set of rooted trees that satisfies the following properties.
\begin{enumerate}
    \item For each $s \in S$, the set contains a tree $T_s = (V_s, E_s)$ rooted at $s$ with $V_s \subseteq X$ and $E_s \subseteq E_X$.
    \item For each $s \in S$, each leaf of $T_s$ is in $S \cup D$.
    \item For each $s_1,s_2 \in S$, $V_{s_1}$ and $V_{s_2}$ are disjoint.
    \item For each $u \in D$, there is a tree $T_s$ such that $u \in V_s$, i.e., $D \subseteq \bigcup_{s \in S} V_s$.
    \item For each $s \in S$ and $u \in V_s$, the unique path from $s$ to $u$ in $T_s$ is a shortest path from $s$ to $u$ in $G_X$, and $s$ has the smallest distance to $u$ among all amoebots in $S$.
\end{enumerate}
We call an $(S,X)$-shortest path forest also an \emph{$S$-shortest path forest}.

We consider the \emph{$(k,\ell)$-shortest path forest problem ($(k,\ell)$-SPF)} for $k, \ell \geq 1$.
Let two sets $S, D \subseteq X$ of amoebots be given such that $|S| = k$ and $|D| = \ell$, i.e., each amoebot $u \in X$ knows whether $u \in S$ and whether $u \in D$.
We say that $X$ computes a $(S,D)$-shortest path forest if each amoebot in $\bigcup_{s \in S} V_s \setminus S$ knows its parent within the $(S,D)$-shortest path forest.
The goal of the amoebot structure is to compute a $(S,D)$-shortest path forest.

Note that we obtain the classical \emph{single pair shortest path problem (SPSP)} for $k = \ell = 1$, and the classical \emph{single source shortest path problem (SSSP)} for $k = 1$ and $\ell = n$.

In the remainder of this paper, we make the following assumptions.
First, the amoebot structure has no holes.
Second, the amoebots agree on a common compass orientation and chirality.
Third, the amoebots agree on a leader, i.e., a unique amoebot.
We can establish the latter two assumptions within $O(\log n)$ rounds w.h.p.\ (see \Cref{sec:coordination}).

Under these preconditions, we will present two deterministic algorithms for $(k,\ell)$-SPF.
Our \emph{shortest path tree algorithm} solves the problem within $O(\log\ell)$ rounds for $k = 1$,
and our \emph{shortest path forest algorithm} within $O(\log n \log^2 k)$ for $k \geq 1$.
Note that the former result implies that we can solve SPSP within $O(1)$ rounds, and SSSP within $O(\log n)$ rounds.

The main challenge to achieve these results was to find the right techniques to cope with the model's limitations (e.g., constant memory) and to fully exploit the model's potential (e.g., fast communication over long distances) because in general, the lower bound for solving the SSSP problem in convential networks is $\Omega(\sqrt n + \mathit{diam})$ \cite{DBLP:journals/siamcomp/Elkin06,DBLP:journals/siamcomp/SarmaHKKNPPW12}.
Therefore, it is very surprising that we found polylogarithmic solutions to $(k,\ell)$-SPF, which again shows the power of the model.



\subsection{Related Work}

The reconfigurable circuit extension was introduced by Feldmann et al.\ \cite{DBLP:journals/jcb/FeldmannPSD22}.
They proposed solutions for the leader election, compass alignment, and chirality agreement problems.
Each of these solutions requires $O(\log n)$ rounds w.h.p.\footnote{An event holds \emph{with high probability (w.h.p.)} if it holds with probability at least $1 - 1/n^c$ where the constant $c$ can be made arbitrarily large.}
Afterwards, they considered the recognition of various classes of shapes.
An amoebot structure is able to detect parallelograms with linear or polynomial side ratio within $O(\log n)$ rounds w.h.p.
Further, an amoebot structure is able to detect shapes composed of triangles within $O(1)$ rounds if the amoebots agree on a chirality.

Feldmann et al.\ proposed the PASC algorithm which allows the amoebot structure to compute distances \cite{DBLP:journals/jcb/FeldmannPSD22,DBLP:conf/dna/PadalkinSW22}.
With the help of it, Padalkin et al.\ were able to solve the global maxima, spanning tree, and symmetry detection problems in polylogarithmic time \cite{DBLP:conf/dna/PadalkinSW22}.
It is also a crucial tool for the results of this paper.
We defer to \Cref{sec:pasc} for more details.

\iftoggle{full}{

Shortest path problems are broadly studied both in the sequential and distributed setting.
In the following, we will discuss the state of the art in various relevant distributed models.
We will limit our considerations to the state of the art algorithms for the exact shortest path problems.
We refer to the cited papers for a more detailed overview of results in the respective models.

}{

Shortest path problems are broadly studied both in the sequential and distributed setting.
In the following, we will discuss the state of the art in various relevant distributed models.
We will limit our considerations to the state of the art algorithms for the exact SSSP and SPSP.
We refer to the cited papers for a more detailed overview of results in the respective models.

}

In the amoebot model, Kostitsyna et al.\ were the first to consider SSSP \cite{DBLP:conf/wdag/KostitsynaPS22,DBLP:conf/wdag/KostitsynaPS23}.
By applying a breadth-first search approach, they compute a shortest path tree within $O(n^2)$ rounds.
For simple amoebot structures without holes, they introduced feather trees -- a special type of shortest path trees.
These can be computed within $O(\mathit{diam})$ rounds.
To our knowledge, there is no further work on shortest path problems in the amoebot model or its reconfigurable circuit extension.

\iftoggle{full}{

In communication networks, adjacent nodes are able to communicate via messages.
The CONGEST model limits the size of each message to a logarithmic number of bits (in $n$).
For weighted SSSP, Chechik and Mukhtar proposed a randomized algorithm that takes $\tilde O(\sqrt n \cdot \mathit{diam}^{1/4} + \mathit{diam})$ rounds \cite{DBLP:journals/dc/ChechikM22}.
The best known lower bound is $\Omega(\sqrt n + \mathit{diam})$ \cite{DBLP:journals/siamcomp/Elkin06,DBLP:journals/siamcomp/SarmaHKKNPPW12}.
It also holds for any approximation factor \cite{DBLP:journals/siamcomp/SarmaHKKNPPW12}.
For weighted APSP, Bernstein and Nanongkai proposed a randomized algorithm that takes $\tilde O(n)$ rounds \cite{DBLP:conf/stoc/BernsteinN19}.
This matches the lower bound up to polylogorithmic factors \cite{DBLP:conf/stoc/LenzenP13,DBLP:conf/stoc/Nanongkai14}.

}{

In communication networks, adjacent nodes are able to communicate via messages.
The CONGEST model limits the size of each message to a logarithmic number of bits (in $n$).
For weighted SSSP, Chechik and Mukhtar proposed a randomized algorithm that takes $\tilde O(\sqrt n \cdot \mathit{diam}^{1/4} + \mathit{diam})$ rounds \cite{DBLP:journals/dc/ChechikM22}.
The best known lower bound is $\Omega(\sqrt n + \mathit{diam})$ \cite{DBLP:journals/siamcomp/Elkin06,DBLP:journals/siamcomp/SarmaHKKNPPW12}.

}

\iftoggle{full}{

In hybrid communication networks \cite{DBLP:conf/soda/AugustineHKSS20}, nodes are able to establish new (global) edges.
Censor-Hillel et al.\ proposed the best known algorithm for weighted SSSP that takes $O(n^{1/3})$ rounds \cite{DBLP:conf/stacs/Censor-HillelLP21}.
Faster algorithms are known for certain classes of graphs:
Feldmann et al.\ solved the problem in $O(\log n)$ rounds for cactus graphs \cite{DBLP:conf/opodis/FeldmannHS20}, and Coy et al.\ solved the problem in $O(\log n)$ rounds for simple grid graphs \cite{DBLP:conf/sirocco/CoyCSSW23}.
The latter make use of portals graphs to compute a partial solution for each dimension, which are then combined.
The portal graphs are another crucial tool for the results of this paper.
We defer to \Cref{sec:portal} for more details.
For weighted APSP, Kuhn and Schneider proposed an $\tilde O(\sqrt n)$ algorithm \cite{DBLP:conf/podc/KuhnS20}.
This matches the lower bound up to polylogorithmic factors \cite{DBLP:conf/soda/AugustineHKSS20}.

}{

In hybrid communication networks \cite{DBLP:conf/soda/AugustineHKSS20}, nodes are able to establish new (global) edges.
Censor-Hillel et al.\ proposed the best known algorithm for weighted SSSP that takes $O(n^{1/3})$ rounds \cite{DBLP:conf/stacs/Censor-HillelLP21}.
Faster algorithms are known for certain classes of graphs:
Feldmann et al.\ solved the problem in $O(\log n)$ rounds for cactus graphs \cite{DBLP:conf/opodis/FeldmannHS20}, and Coy et al.\ solved the problem in $O(\log n)$ rounds for simple grid graphs \cite{DBLP:conf/sirocco/CoyCSSW23}.
The latter make use of portals graphs to compute a partial solution for each dimension, which are then combined.
The portal graphs are another crucial tool for the results of this paper.
We defer to \Cref{sec:portal} for more details.

}

\iftoggle{full}{

In the beeping model \cite{DBLP:conf/wdag/CornejoK10}, each node is either in beeping or listening mode.
If a node is in beeeping mode, it sends a beep to all its neighbors.
If a node is in listening mode, it perceives beeps in its neighborhood.
Just as in the reconfigurable circuit extension, it neither knows the origin nor the number of origins.
The model differs from the reconfigurable circuit extension in that nodes cannot establish circuits and can only listen to their neighborhoods.
Dufoulon et al.\ proposed two algorithms for shortest path problems \cite{DBLP:conf/innovations/DufoulonEG23}.
The first solves SPSP within $O(\mathit{diam} \log\log n + \log^3 n)$ rounds w.h.p., and the second solves SSSP within $O(\mathit{diam} \log^2 n + \log^3 n)$ rounds w.h.p.
To our knowledge, there is no further work on shortest path problems in the beeping model.

}{

In the beeping model \cite{DBLP:conf/wdag/CornejoK10}, each node is either in beeping or listening mode.
If a node is in beeeping mode, it sends a beep to all its neighbors.
If a node is in listening mode, it perceives beeps in its neighborhood.
Just as in the reconfigurable circuit extension, it neither knows the origin nor the number of origins.
The model differs from the reconfigurable circuit extension in that nodes cannot establish circuits and can only listen to their neighborhoods.
Dufoulon et al.\ proposed two algorithms for shortest path problems \cite{DBLP:conf/innovations/DufoulonEG23}.
The first solves SPSP within $O(\mathit{diam} \log\log n + \log^3 n)$ rounds w.h.p., and the second solves SSSP within $O(\mathit{diam} \log^2 n + \log^3 n)$ rounds w.h.p.

}



\section{Preliminaries}

In this section, we discuss previous results that we utilize to achieve the results of this paper.
The first subsection deals with the coordination of amoebots.
In the second subsection, we present the PASC algorithm, and in the third subsection, we discuss portal trees.


\subsection{Coordination}
\label{sec:coordination}

One of the difficulties in designing algorithms with reconfigurable circuits is the coordination of the amoebots.
In order to simplify the coordination, we have made the following two assumptions.
First, the amoebots have to agree on a common compass orientation and chirality.
Second, the amoebots agree on a leader, i.e., a unique amoebot.
If these assumptions are not satisfied, we can establish them in a preprocessing phase.
For that, we make use of the following results by Feldmann et al.

\begin{theorem}[Feldmann et al.\ \cite{DBLP:journals/jcb/FeldmannPSD22}]
    There is an algorithm that aligns all compasses and chiralities within $O(\log n)$ rounds w.h.p.
\end{theorem}

\begin{theorem}[Feldmann et al.\ \cite{DBLP:journals/jcb/FeldmannPSD22}]
    There is an algorithm that elects a leader within $\Theta(\log n)$ rounds w.h.p.
\end{theorem}

Hence, the preprocessing phase requires $O(\log n)$ rounds w.h.p.
Note that we can omit the leader election for $(1, \ell)$-SPF since we can simply elect the only source as the leader.
Also, note that while this preprocessing phase is randomized, all presented algorithms in this paper are deterministic.

\iftoggle{full}{
Moreover, our algorithms apply some primitives on several subsets of amoebots in parallel.
These might take different amounts of rounds for each subset.
In order to synchronize the amoebot structure, we apply the synchronization technique by Padalkin et al.
We refer to \cite{DBLP:conf/dna/PadalkinSW22} for more details.
}{}



\subsection{PASC Algorithm}
\label{sec:pasc}

One of the essential primitives used in this paper is the \emph{primary and secondary circuit algorithm (PASC algorithm)} that was introduced by Feldmann et al.~\cite{DBLP:journals/jcb/FeldmannPSD22}.

\iftoggle{full}{

\begin{lemma}[Padalkin et al.~\cite{DBLP:conf/dna/PadalkinSW22}]
\label{lem:pasc:chain}
    Let a chain of $m$ amoebots be given, i.e., each amoebot knows its predecessor and successor.
    The PASC algorithm computes the distance of each amoebot to the first amoebot bit by bit.
    More precisely, in the $i$-th iteration of the PASC algorithm, each amoebot computes the $i$-th bit of its distance to the first amoebot.
\end{lemma}

\begin{lemma}[Feldmann et al.~\cite{DBLP:journals/jcb/FeldmannPSD22}]
\label{lem:pasc:runtime}
    Each iteration of the PASC algorithm requires two rounds.
    The PASC algorithm (performed on a chain of $m$ amoebots) terminates after $O(\log m)$ iterations.
\end{lemma}

}{

\begin{lemma}[Feldmann et al.~\cite{DBLP:journals/jcb/FeldmannPSD22}, Padalkin et al.~\cite{DBLP:conf/dna/PadalkinSW22}]
\label{lem:pasc:chain}
    Let a chain of $m$ amoebots be given, i.e., each amoebot knows its predecessor and successor.
    The PASC algorithm computes the distance of each amoebot to the first amoebot bit by bit.
    More precisely, in the $i$-th iteration of the PASC algorithm, each amoebot computes the $i$-th bit of its distance to the first amoebot.
    Furthermore, the PASC algorithm terminates after $O(\log m)$ rounds.
\end{lemma}

}

\iftoggle{full}{
We will use the PASC algorithm as a black box.
We refer to \cite{DBLP:conf/dna/PadalkinSW22} for the details.
We can adapt the PASC algorithm to compute distances in tree structures and prefix sums along the chain as shown in the following corollaries.
}{
We will use the PASC algorithm as a black box.
We refer to \cite{DBLP:conf/dna/PadalkinSW22} for the details.
We can adapt the PASC algorithm to compute distances in tree structures.
}

\begin{corollary}
\label{cor:pasc:tree}
    Let a rooted tree of amoebots with height $h$ be given, i.e., each amoebot knows its parent and its children.
    The PASC algorithm computes the distance of each amoebot to the root bit by bit.
    More precisely, in the $i$-th iteration of the PASC algorithm, each amoebot computes the $i$-th bit of its distance to the root.
    Furthermore, the PASC algorithm terminates after $O(\log h)$ rounds.
\end{corollary}

\iftoggle{full}{
\begin{proof}
    We simply apply \Cref{lem:pasc:chain} simultaneously on each path from the root to each leaf.
    Each amoebot can reuse its partition sets for all paths such that we still only need two external links for each edge of the tree.
    We refer to \cite{DBLP:conf/dna/PadalkinSW22} for similar extensions.
    
    The runtime depends on the longest path, i.e., on the height of the tree.
    By \Cref{lem:pasc:runtime}, the PASC algorithm terminates after $O(\log h)$ rounds.
\end{proof}
}{}

\iftoggle{full}{
\begin{corollary}
\label{cor:pasc:prefixsum}
    Let a chain $(v_0, \dots, v_{m-1})$ of $m$ amoebots be given, i.e., each amoebot knows its predecessor and successor.
    Let $V$ denote the set of all amoebots in the chain.
    Further, let a weight function $\weight : V \to \{ 0, 1 \}$ be given, i.e., each amoebot knows its weight.
    The PASC algorithm computes the prefix sum $\prefix_{v_i} = \sum_{j=0}^i \weight(v_i)$ of each amoebot $v_i$ bit by bit.
    More precisely, in the $i$-th iteration of the PASC algorithm, each amoebot computes the $i$-th bit of its prefix sum.
    Furthermore, the PASC algorithm terminates after $O(\log W)$ rounds where $W = \sum_{j=0}^{m-1} \weight(v_i)$.
\end{corollary}
}{}

\iftoggle{full}{
\begin{proof}
    We first append a virtual amoebot $s$ with $\weight(s) = 0$ to the start of the chain, which is simulated by $v_0$.
    Then, we apply the PASC algorithm.
    However, we only let amoebots with weight $1$ participate while all other amoebots simply forward the signals of their predecessors to their successors.
    By \Cref{lem:pasc:chain}, each amoebot with weight $1$ computes its weighted distance to $s$, which is equal to its prefix sum.
    Each amoebot $v \in V$ with weight $0$ is able to read the forwarded signals.
    This allows it to determine the value computed by the last amoebot with weight $1$ on the subchain from $s$ to $v$, which is equal to its prefix sum.
    
    The runtime does only depend on the number of participating amoebots.
    There are exactly $W$ many amoebots with weight $1$.
    Hence, by \Cref{lem:pasc:runtime}, the algorithm requires $O(\log W)$ rounds.
\end{proof}
}{}



\subsection{Portal Graph}
\label{sec:portal}

Coy et al.~\cite{DBLP:conf/sirocco/CoyCSSW23} have solved the shortest path problem for hybrid communication networks that can be modelled as grid graphs without holes.
For that, they have made use of portal graphs.

\begin{definition}[Coy et al.~\cite{DBLP:conf/sirocco/CoyCSSW23}]
\label{def:portal}
    Let $G = (V, E)$ be a connected subgraph of a square grid.
    Let $E_x \subseteq E$ be the set of edges parallel to the $x$-axis.
    We call the connected component of $(V, E_x)$ \emph{$x$-portals}.
    For each $u \in V$, let $\portal_x(u)$ denote the portal that contains $u$.
    Two portals $P_1$ and $P_2$ are adjacent iff there exists an edge $(v_1,v_2) \in E$ such that $v_1 \in P_1$ and $v_2 \in P_2$.
    We define \emph{$y$-portals} analogously.
\end{definition}

\begin{definition}[Coy et al.~\cite{DBLP:conf/sirocco/CoyCSSW23}]
\label{def:portal_graph}
    The \emph{$x$-portal graph} $\mathcal P_x$ is the graph with vertices corresponding to the $x$-portals.
    Two vertices of $\mathcal P_x$ are adjacent iff the corresponding portals are adjacent.
    We define the portal graph $\mathcal P_y$ analogously.
\end{definition}

We may omit the axis in the notation if it is arbitrary or clear from the context.

\begin{lemma}[Coy et al.~\cite{DBLP:conf/sirocco/CoyCSSW23}]
\label{lem:portal_graph:tree}
    All portal graphs are trees if the grid graph has no holes.
\end{lemma}

Let $\dist(u,v)$ denote the distance between $u$ and $v$ in $G$, and $\dist_d(u,v)$ denote the distance between $\portal(u)$ and $\portal(v)$ in $\mathcal P_d$.
Further, let $\dist(U,v) = \min_{u \in U} \dist(u,v)$.

\begin{lemma}[Coy et al.~\cite{DBLP:conf/sirocco/CoyCSSW23}]
\label{lem:portal_graph:square}
    Let $G$ be a square grid graph without holes.
    Then, $\dist(u,v) = \dist_x(u,v) + \dist_y(u,v)$ holds.
\end{lemma}

We adapt that definitions and results to triangular grids.
For that, we extend the definitions by \emph{$z$-portals} and \emph{$z$-portal graph} analogously (see \Cref{fig:portal:axes}).
\Cref{lem:portal_graph:tree} and its proof still hold for triangular grids.
But we have to adapt \Cref{lem:portal_graph:square} as follows.

\begin{lemma}
\label{lem:portal_graph:triangular}
    Let $G$ be a triangular grid graph without holes.
    Then, $2 \cdot \dist(u,v) = \dist_x(u,v) + \dist_y(u,v) + \dist_z(u,v)$ holds.
\end{lemma}

\iftoggle{full}{
\begin{proof}
    Intuitively, the sum of the distances in the portal graphs count each edge twice.
    The following proof is analogous to the proof of Lemma 21 in \cite{DBLP:conf/sirocco/CoyCSSW23}.
    Consider a shortest path $(u=w_1,\dots,w_m=v)$ from $u$ to $v$.
    We prove the statement for this path by induction.
    The induction base $i = 1$ holds trivially. 
    
    Suppose that the statement holds for the first $i$ nodes of the path.
    Consider node $w_{i+1}$.
    By induction hypothesis, $2 \cdot \dist(u,w_i) = \dist_x(u,w_i) + \dist_y(u,w_i) + \dist_z(u,w_i)$ holds.
    
    For now, assume that edge $(w_i,w_{i+1})$ is parallel to the $x$-axis.
    This implies that $w_i$ and $w_{i+1}$ belong to the same $x$-portal and hence, $\dist_x(u,w_{i+1}) = \dist_x(u,w_i)$ holds.
    Note that shortest path cannot visit any portals more than once.
    Otherwise, we would be able to shorten the path.
    Combined with \Cref{lem:portal_graph:tree}, this implies $\dist_y(u,w_{i+1}) = \dist_y(u,w_i) + 1$ and $\dist_z(u,w_{i+1}) = \dist_z(u,w_i) + 1$.
    Altogether, we have $2 \cdot \dist(u,w_{i+1}) =  2 \cdot \dist(u,w_i) + 2 = \dist_x(u,w_i) + \dist_y(u,w_i) + \dist_z(u,w_i) + 2 = \dist_x(u,w_{i+1}) + \dist_y(u,w_{i+1}) + \dist_z(u,w_{i+1})$.
    The cases where $(w_i,w_{i+1})$ is parallel to the $y$ or $z$-axis work analogously.
\end{proof}
}{}

However, the amoebot structure has no access to the portal graphs.
Instead, similar to \cite{DBLP:conf/sirocco/CoyCSSW23}, we define a subgraph of $G_X$ for each portal graph that preserves its properties as follows.

\begin{definition}
    The \emph{implicit portal graph} of $\mathcal P_x$ is the subgraph $T = (X,E_x \cup E'_x)$ of $G_X$ where $E_x$ is the set of edges parallel to the $x$-axis and $E'_x$ cotains the westernmost edge between each pair of adjacent $x$-portals (see \Cref{fig:portal:full,fig:portal:x}).
    We define the implicit portal graphs of $\mathcal P_y$ and $\mathcal P_z$ analogously (see \Cref{fig:portal:y,fig:portal:z}).
\end{definition}

Intuitively, the vertices of a portal graph are portals while the vertices of an implicit portal graph are amoebots.
Note that $T$ is a tree since we connect all amoebots of each portal into a chain and each portal graph is a tree.

\iftoggle{full}{

Furthermore, note that each amoebot can locally decide which of its incident edges belong to $T$.
For example, consider an amoebot for the implicit portal graph of the $x$-portal graph.
The edges to the west and east always belong to $T$ since they are parallel to the $x$-axis, i.e., they belong to $E_x$.
The edge to the north-west (south-west) belongs to $T$ if there is no edge to the west since the amoebot is the westernmost amoebot of its portal.
The edge to the north-east (south-east) belongs to $T$ if there is no edge to the north-west (south-west) since the amoebot to the north-east (south-east) is the westernmost amoebot of its portal.

}{}

\begin{figure}[tbp]
    \begin{minipage}[t]{.21\linewidth}
        \centering
        \includegraphics{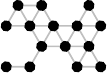}
        \subcaption{Initial structure.}
        \label{fig:portal:full}
    \end{minipage}
    \hfill
    \begin{minipage}[t]{.21\linewidth}
        \centering
        \includegraphics{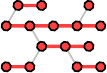}
        \subcaption{$x$-portal graph.}
        \label{fig:portal:x}
    \end{minipage}
    \hfill
    \begin{minipage}[t]{.21\linewidth}
        \centering
        \includegraphics{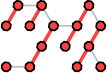}
        \subcaption{$y$-portal graph.}
        \label{fig:portal:y}
    \end{minipage}
    \hfill
    \begin{minipage}[t]{.21\linewidth}
        \centering
        \includegraphics{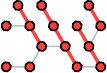}
        \subcaption{$z$-portal graph.}
        \label{fig:portal:z}
    \end{minipage}
    \hfill
    \begin{minipage}[t]{.13\linewidth}
        \centering
        \includegraphics{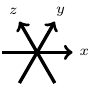}
        \subcaption{Axes.}
        \label{fig:portal:axes}
    \end{minipage}
    \caption{
        (Implicit) portal graphs.
        Each red connected component indicates a portal.
        We obtain the portal graphs by fusing the amoebots of each portal to a single node.
    }
    \label{fig:portal}
\end{figure}

\begin{lemma}
\label{obs:portal}
    Let $P$ be a portal.
    Let $u,v \in V \setminus P$ be two nodes.
    The shortest path from $u$ to $v$ traverses $P$ iff $u$ and $v$ are not in the same connected component of $V \setminus P$.
\end{lemma}

\iftoggle{full}{
\begin{proof}
    Suppose that the shortest path from $u$ to $v$ traverses $P$.
    Let $(u, \dots, v)$ denote the sequence of amoebots traversed by the shortest path.
    Consider a subpath $(w_1, \dots, w_m)$.
    If $w_1 \in P'$ and $w_m \in P'$, then for each $1 \leq i \leq m$, $w_i \in P'$.
    Otherwise, we could shorten the shortest path with the path between $w_1$ and $w_m$ in $P'$, which contradicts the assumption.
    Now, suppose that $u$ and $v$ are in the same connected component of $V \setminus P$.
    Let $P'$ be the portal in the connected component of $u$ and $v$ that is adjacent to $P$.
    Note that $P'$ is unique since we assume that the amoebot structure has no holes.
    Then, the subpath from $u$ to $p$ has to traverse an amoebot $u' \in P'$, and the subpath from $p$ to $v$ has to traverse an amoebot $v' \in P'$, respectively.
    Consider the subpath $(u', \dots, p, \dots, v')$.
    Since $u' \in P'$ and $v' \in P'$, $p \in P'$.
    This is a contradiction to $p \in P$.
    Hence, $u$ and $v$ cannot be in the same connected component of $V \setminus P$ if the shortest path from $u$ to $v$ traverses $P$.
    
    Suppose that $u$ and $v$ are not in the same connected component of $V \setminus P$.
    Then, each path from $u$ to $v$ has to traverse $P$.
\end{proof}
}{}



\iftoggle{full}{

\section{Tree Primitives}

Our algorithms make extensive use of tree structures.
In this section, we introduce important tree primitives that we believe to be of independent interest.
These are not limited to the geometric variant of the amoebot model.
Therefore, we will speak of nodes instead of amoebots.

In the first subsection, we will introduce the Euler tour technique by Tarjan and Vishkin and adapt it to the reconfigurable circuit extension.
This technique serves as the framework for the tree primitives that we present in the second and third subsection.
The first primitive roots the tree at a given node and prunes any subtrees without nodes in a given set.
The second primitive elects a node from a given set of nodes.
The third primitive computes the centroid(s) of a tree.
In the last subsection, we will adapt the primitives to implicit portal graphs.


\subsection{Euler Tour Technique}
\label{sec:euler}

In order to solve both problems, we will utilize the \emph{Euler tour technique (ETT)} that was introduced by Tarjan and Vishkin for the \emph{PRAM model}\footnote{More precisely, Tarjan and Vishkin assume the concurrent-read, concurrent-write PRAM (CRCW PRAM) model.} \cite{DBLP:journals/siamcomp/TarjanV85}.
The technique allows the computation of various tree functions, e.g., computing a rooted version of a tree, a pre- and postorder numbering of the nodes, the number of descendants of each node, the level of each node, and the centroid(s) of a tree \cite{DBLP:conf/icpp/CongB04,DBLP:journals/siamcomp/TarjanV85}.
In the following, we will explain the technique and adapt it to the reconfigurable circuit extension.

Let $T = (V_T, E_T)$ be a tree and $r \in V_T$ an arbitrary node.
We replace each undirected edge $\{ u, v\} \in E_T$ with two directed edges $(u,v)$ and $(v,u)$.
Let $T' = (V_T, E'_T)$ denote the resulting graph.
Consider any Euler cycle of $T'$, e.g., for each edge $(u,v) \in E'_T$, let edge $(v,w) \in E'_T$ be the next edge of the Euler cycle where $w$ is the next counterclockwise neighbor of $v$ with respect to $u$.
By splitting the Euler cycle at $r$, we obtain an Euler tour $\pi$ of $T'$ that starts and ends at $r$.

Let $\weight : E'_T \to \{ 0, 1 \}$ be any weight function.
The exact definition of the function is application-specific.
Let $(e_0, \dots, e_{|E'_T|-1})$ denote the sequence of edges traversed by $\pi$.
The first step of the ETT is to compute the prefix sum of each edge, i.e., each $e_i \in E'_T$ computes $\prefix_{e_i} = \sum_{j = 0}^i \weight(e_i)$.
For that, Tarjan and Vishkin utilize a processor for each edge, and apply a doubling technique, which requires pointers between non-incident edges (see \cite{DBLP:journals/siamcomp/TarjanV85} for more details).
Neither is possible in the amoebot model.
Nor do we have the memory to store the prefix sums.

Instead, we will utilize the PASC algorithm on the nodes to compute the prefix sums bit by bit as follows.
Each node in $V_T$ operates an independent instance for each of its occurrences on $\pi$.
Let $(v_0, \dots, v_{|E'_T|})$ denote the sequence of instances traversed by $\pi$.
Observe that node $r$ operates the first and last instance of the sequence, i.e., $v_0$ and $v_{|E'_T|}$.
Let $V'_T$ denote the set of all instances.

We define a weight function $\weight : V'_T \to \{ 0, 1 \}$ such that $\weight(v_i) = \weight(e_i)$ for each $0 \leq i < |E'_T|$ and $\weight(v_{|E'_T|}) = 0$.
We apply the PASC algorithm on the sequence of nodes to compute the prefix sums of each node bit by bit, i.e., each $v_i \in V'_T$ computes $\prefix_{v_i} = \sum_{j = 0}^i \weight(v_i)$ bit by bit.
By definition, $\prefix_{e_i} = \prefix_{v_i} = \prefix_{v_{i+1}} - \weight(v_{i+1})$ holds.
Hence, for each edge $e_i = (v_i, v_{i+1}) \in E'_T$, both incident nodes $v_i$ and $v_{i+1}$ are able to compute $\prefix_{e_i}$ bit by bit.

In the second step of the ETT, each node $u \in V_T$ computes $\prefix_{(u,v)} - \prefix_{(v,u)}$ for each of its neighbors $v \in \neighborhood(u)$.
As in the first step, we have to compute all differences bit by bit.
This can be done in parallel with the first step.

\begin{lemma}
\label{lem:euler}
    Let an Euler tour $\pi$ of $T' = (V_T, E'_T)$ be given, i.e., each amoebot knows its predecessor and successor of each of its occurrences.
    Further, let a weight function $\weight : E'_T \to \{ 0, 1 \}$ be given, i.e., for each $(u, v) \in E'_T$, amoebot $u$ knows $\weight(u,v)$.
    The ETT computes the following values in parallel.
    \begin{enumerate}
        \item 
        Each amoebot computes the prefix sum of each of its incident edges.
        More precisely, in the $i$-th iteration of the ETT, each amoebot computes the $i$-th bit of the prefix sum of each of its incident edges.
        \item 
        Each amoebot $u \in V_T$ computes $\prefix_{(u,v)} - \prefix_{(v,u)}$ for each of its neighbors $v \in \neighborhood(u)$.
        More precisely, in the $i$-th iteration of the ETT, each amoebot $u \in V_T$ computes the $i$-th bit of $\prefix_{(u,v)} - \prefix_{(v,u)}$ for each of its neighbors $v \in \neighborhood(u)$.
    \end{enumerate}
    Furthermore, the ETT terminates after $O(\log W)$ rounds where $W = \sum_{e \in E'_T} \weight(e)$.
\end{lemma}

\begin{proof}
    The statement follows directly from \Cref{cor:pasc:prefixsum}.
\end{proof}

\begin{corollary}
\label{cor:euler:number}
    In particular, amoebot $r$ computes $W$ bit by bit.
\end{corollary}

\begin{proof}
    Recall that $r$ operates the last instance $v_{|E'_T|}$.
    By definition, $\prefix_{v_{|E'_T|}} = W$.
\end{proof}

\begin{remark}
    In the first step of the ETT, each instance only requires $O(1)$ memory space.
    Each node $u \in V_T$ operates $\Theta(\degree_T(u))$ instances where $\degree_T(v)$ denotes the degree of node $v$ within $T$.
    In the second step of the ETT, each amoebot $u \in V_T$ computes a difference for each of its $\degree_T(u)$ neighbors.
    Each computation requires $O(1)$ memory space.
    Note that in the geometric variant of amoebot model, amoebots have sufficient memory since their degree is bounded by $6$.
\end{remark}

In the remainder of this section, we will define a weight function $\weight_Q : E'_T \to \{ 0, 1 \}$ for a set $Q \subseteq V_T$ that we will use in the subsequent sections.
Each node $u \in Q$ marks exactly one of its out-going edges.
For each $e \in E'_T$, we set $\weight_Q(e) = 1$ if $e$ has been marked, and $\weight_Q(e) = 0$ otherwise.
Note that $W = |Q|$.
For this weight function, we obtain the following properties, which are a generalization of the results in \cite{DBLP:journals/siamcomp/TarjanV85}.

\begin{lemma}
\label{lem:euler:properties}
    Let $u \in V_T$.
    If $p \in \neighborhood(u)$ is the parent of $u$ with respect to $r$, the following properties hold.
    \begin{enumerate}
        \item The subtree of $u$ with respect to $r$ contains $\prefix_{(u,p)} - \prefix_{(p,u)}$ nodes in $Q$.
        \item $\prefix_{(u,p)} - \prefix_{(p,u)} \geq 0$.
    \end{enumerate}
    If $c \in \neighborhood(u)$ is a child of $u$ with respect to $r$, the following properties hold.
    \begin{enumerate}
        \setcounter{enumi}{2}
        \item The subtree of $c$ with respect to $r$ contains $\prefix_{(c,u)} - \prefix_{(u,c)}$ nodes in $Q$.
        \item $\prefix_{(u,c)} - \prefix_{(c,u)} \leq 0$.
    \end{enumerate}
\end{lemma}

\begin{proof}
    The Euler tour $\pi$ enters the subtree of $u$ through $(p,u)$, traverses all edges of the subtree, and leaves the subtree through $(u,p)$.
    By definition, the traversed subpath contains $\prefix_{(u,p)} - \prefix_{(p,u)}$ instances in $Q$.
    This is also the number of nodes in $Q$ in the subtree of $u$ since the subpath consists of all instances of the subtree and each node in $Q$ marks exactly one instance.
    This already proves the first two properties.
    Note that they immediately imply the other two properties.
\end{proof}

\begin{figure}[tbp]
    \begin{minipage}[t]{\textwidth}
        \centering
        \includegraphics{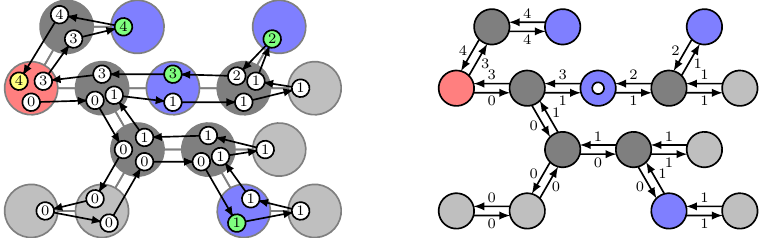}
        \subcaption{Tree structure.}
        \label{fig:rap:tree}
    \end{minipage}

    \bigskip
    
    \begin{minipage}[t]{\textwidth}
        \centering
        \includegraphics{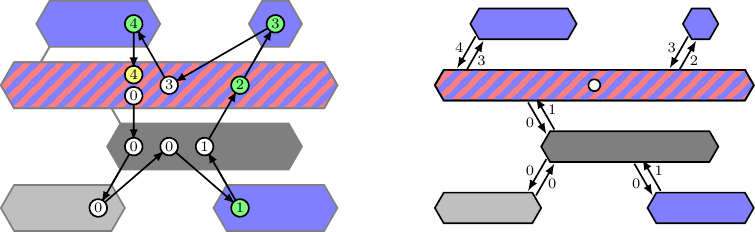}
        \subcaption{Portal graph.}
        \label{fig:rap:pg}
    \end{minipage}

    \bigskip
    
    \begin{minipage}[t]{\textwidth}
        \centering
        \includegraphics{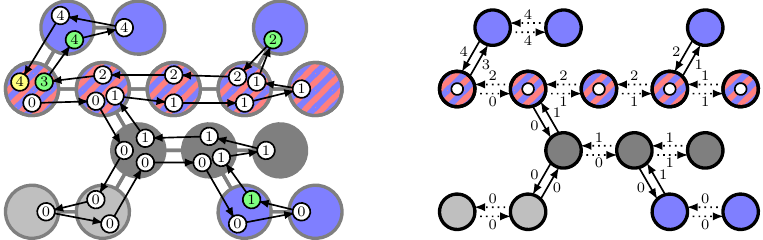}
        \subcaption{Implicit portal graph.}
        \label{fig:rap:ipg}
    \end{minipage}
    \caption{
        Tree algorithms.
        The red nodes/portals indicates $r$ and $R$, respectively.
        The blue nodes/portals indicate $Q$ and $\mathcal Q$, respectively.
        Note that in the (implicit) portal graph, $R \in \mathcal Q$.
        The left figures show the PASC algorithm on the Eulerian path, respectively.
        The smaller nodes indicate the instances of the bigger nodes/portals.
        Only the green instances participate in the PASC algorithm, i.e., they have weight $1$.
        The yellow node indicates the last instance which computes $m$.
        The numbers within the instances indicate their prefix sums.
        The right figures show the prefix sums for the edges, respectively.
        The root and prune algorithm prunes the light gray nodes/portals.
        The centroid algorithm identifies nodes/portals marked by a white dot as $Q$- or $\mathcal Q$-centroids.
    }
    \label{fig:rap}
\end{figure}


\subsection{Root and Prune Primitive}
\label{sec:rap}

Let $T = (V_T, E_T)$ be a tree.
Let a node $r \in V_T$ and a subset $Q \subseteq V_T$ be given, i.e., each node $u \in V_T$ knows whether $u = r$ and whether $u \in Q$.
Let $\subtree_r(u)$ denote the set of nodes in the subtree of $u$ in the tree rooted at $r$.
Let $V_Q = \{u \in V_T \mid \subtree_r(u) \cap Q \neq \emptyset \}$ denote the set of nodes whose subtree contains nodes in $Q$.
The goal of the \emph{root and prune problem} is to root $T$ at $r$ and to prune all subtrees without a node in $Q$, i.e., (i) each node $u \in V_T$ determines whether $u \in V_Q$, and (ii) each amoebot $u \in V_Q \setminus \{ r \}$ identifies its parent $v \in \neighborhood(u)$ with respect to $r$.
In order to solve this problem, we make use of the following \Cref{cor:rap:properties} of \Cref{lem:euler:properties} and \Cref{lem:rap:root}.

\begin{corollary}
\label{cor:rap:properties}
    Let $u \in V_T \setminus \{ r \}$.
    The subtree of $u$ with respect to $r$ does not contain any nodes in $Q$ iff $\prefix_{(u,v)} - \prefix_{(v,u)} = 0$ for all $v \in \neighborhood(u)$.
    Suppose that the subtree of $u$ contains at least one node in $Q$.
    A node $v \in \neighborhood(u)$ is the parent of $u$ iff $\prefix_{(u,v)} - \prefix_{(v,u)} > 0$.
\end{corollary}

\begin{proof}
    The statement follows directly from \Cref{lem:euler:properties}.
\end{proof}

\begin{lemma}
\label{lem:rap:root}
    The subtree of $r$ with respect to $r$ does not contain any nodes in $Q$ iff $|Q| = 0$.
\end{lemma}

\begin{proof}
    The statement holds trivially since the subtree of $r$ is equal to $T$.
\end{proof}

Hence, our \emph{root and prune primitive} solves the problem as follows.
We apply the ETT with $\weight_Q$ as weight function.
In parallel, each amoebot $u \in V_T$ compares $\prefix_{(u,v)} - \prefix_{(v,u)}$ with $0$ for all $v \in \neighborhood(u)$.
Additionally, amoebot $r$ compares $|Q|$ with $0$.
Each amoebot $u \in V_T \setminus \{ r \}$ is in $V_Q$ if there is a $v \in \neighborhood(u)$ such that $\prefix_{(u,v)} - \prefix_{(v,u)} \neq 0$.
Amoebot $r$ is in $V_Q$ if $|Q| > 0$.
Each $u \in V_Q$ identifies its neighbor $v \in \neighborhood(u)$ with $\prefix_{(u,v)} - \prefix_{(v,u)} > 0$ as its parent.

\begin{lemma}
\label{lem:rap}
    Our root and prune primitive roots $T$ at $r$ and prunes all subtrees without a node in $Q$ within $O(\log |Q|)$ rounds.
\end{lemma}

\begin{proof}
    By \Cref{lem:euler}, each amoebot $u \in V_T$ computes $\prefix_{(u,v)} - \prefix_{(v,u)}$ for each of its neighbors $v \in \neighborhood(u)$.
    Recall that $W = |Q|$.
    By \Cref{cor:euler:number}, $r$ computes $|Q|$.
    Hence, by \Cref{cor:rap:properties} and \Cref{lem:rap:root}, each amoebot $u \in V_T$ is able to determine whether $u \in V_Q$, and each amoebot $u \in V_Q \setminus \{ r \}$ is able to identify its parent $v \in \neighborhood(u)$ with respect to $r$.
    The runtime follows from \Cref{lem:euler}.
\end{proof}


\subsection{Election Primitive}

Let $T = (V_T, E_T)$ be a tree.
Let a node $r \in V_T$ and a non-empty subset $Q \subseteq V_T$ be given, i.e., each node $u \in V_T$ knows whether $u = r$ and whether $u \in Q$.
The goal of the \emph{election problem} is to elect a single node $r' \in Q$, i.e., each node $u \in V_T$ knows whether $u = r'$.
Note that this is not a leader election since we already assume that a unique node $r$ is given.

Our \emph{election primitive} solves the problem as follows.
The idea is to apply the ETT with weight function $\weight_Q$ to identify the first marked edge $(u,v) \in E'_T$ on the Euler tour of $T' = (V_T, E'_T)$.
Then, we simply elect $u$.
However, since we are not interested in the prefix sums, we simplify the ETT as follows.
First, we remove all marked edges.
This splits the Euler tour into subpaths.
Note that the first subpath starts at $r$ and ends at $u$.
Each subpath establishes a circuit along the subpath.
Then, node $r$ beeps on the circuit of the first subpath.
We elect the node at the other end.

\begin{lemma}
\label{lem:election}
    Our election primitive elects a single node $r' \in Q$ within $O(1)$ rounds. 
\end{lemma}

\begin{proof}
    The correctness follows from the fact that the first subpath starts at $r$ and ends at $u$.
    Further, the primitive only utilizes a single round.
\end{proof}


\subsection{Centroid Decomposition}
\label{sec:centroid}

Let $T = (V_T, E_T)$ be a tree.
Let a node $r \in V_T$ and a subset $Q \subseteq V_T$ be given, i.e., each node $u \in V_T$ knows whether $u = r$ and whether $u \in Q$.
%
A node $u \in Q$ is a \emph{$Q$-centroid} iff after the removal of $u$, the tree splits into connected components with at most $|Q|/2$ nodes in $Q$, respectively.
The goal of the \emph{$Q$-centroid problem} is to compute the $Q$-centroid(s) of $T$, i.e., each node determines whether it is a $Q$-centroid.
In order to solve this problem, we make use of the following corollary of \Cref{lem:euler:properties}.

\begin{corollary}
\label{cor:centroid:properties}
    Let $u \in V_T$.
    After the removal of $u$, the connected component of $v \in \neighborhood(u)$ contains $|Q| - (\prefix_{(u,v)} - \prefix_{(v,u)})$ nodes in $Q$ if $v$ is the parent of $u$ with respect to $r$, and $\prefix_{(v,u)} - \prefix_{(u,v)}$ nodes in $Q$ if $v$ is a child of $u$ with respect to $r$.
\end{corollary}

\begin{proof}
    The statement follows directly from \Cref{lem:euler:properties}.
\end{proof}

Hence, our \emph{$Q$-centroid primitive} computes the $Q$-centroid(s) as follows.
First, we apply the ETT with $\weight_Q$ as weight function to compute the parents of each amoebot in $Q$ with respect to $r$.
Then, we apply the ETT with $\weight_Q$ as weight function again.
However, after each iteration of the ETT, amoebot $r$ broadcasts the current bit of $|Q|$.
In parallel to the ETT, each amoebot $u \in V_Q \setminus \{ r \}$ computes $|Q| - (\prefix_{(u,p)} - \prefix_{(p,u)})$ where $p$ denotes the parent of $u$.
Additionally, each amoebot $u \in V_Q$ computes $|Q|/2$ and compares it with $\size_u(v)$ for each $v \in \neighborhood(u)$ where $\size_u(v) = |Q| - (\prefix_{(u,v)} - \prefix_{(v,u)})$ if $v$ is the parent of $u$ and $\size_u(v) = \prefix_{(v,u)} - \prefix_{(u,v)}$ otherwise.
An amoebot $u \in Q$ identifies as a $Q$-centroid if $\size_u(v) \leq |Q|/2$ for each $v \in \neighborhood(u)$.

\begin{lemma}
\label{lem:centroid}
    Our $Q$-centroid primitive computes the $Q$-centroid(s) within $O(\log |Q|)$ rounds.
\end{lemma}

\begin{proof}
    By \Cref{lem:rap}, each amoebot in $Q$ is able to identify its parent.
    By \Cref{lem:euler}, each amoebot $u \in V_T$ computes $\prefix_{(u,v)} - \prefix_{(v,u)}$ for each of its neighbors $v \in \neighborhood(u)$.
    Recall that $W = |Q|$.
    By \Cref{cor:euler:number}, $r$ computes $|Q|$, which it broadcasts.
    Hence, by \Cref{cor:centroid:properties}, each amoebot is able to determine whether it is a $Q$-centroid.
    The runtime follows from \Cref{lem:euler}.
\end{proof}


Next, we discuss the existence of $Q$-centroid(s).
A node $u \in Q$ is a \emph{centroid} iff after the removal of $u$, the tree splits into connected components with at most $|V_T|/2$ nodes, respectively.
Let $w : V_T \to \mathbb N_0$ be any weight function.
A node $u \in V_T$ is a \emph{weighted centroid} iff after the removal of $u$, the tree splits into connected components whose nodes have a cumulative weight of at most $W/2$, respectively, where $W = \sum_{v \in V_T} \weight(v)$.
The following results are known for (weighted) centroids.

\begin{theorem}[Jordan \cite{jordan1869assemblages}]
\label{th:centroid:simple}
    Each tree has one or two adjacent centroids.
\end{theorem}

\begin{theorem}[Bielak and Panczyk \cite{DBLP:journals/umcs/BielakP12}]
\label{th:centroid:weighted}
    Each tree has at least one weighted centroid.
    Further, with only positive weights, each tree has one or two adjacent weighted centroids.
\end{theorem}

Note that the existence of weighted centroids for weight function $\weight_Q$ does not imply the existence of $Q$-centroids since these might not be in $Q$.
In fact, a tree might not have any $Q$-centroids, e.g., any graph with $Q = \emptyset$.
However, in the following, we will show that we can augment each non-empty set $Q$ with a set $A_Q$ with $O(|Q|)$ nodes to guarantee the existence of one or two $(Q \cup A_Q)$-centroids.

Let $T = (V_T, E_T)$ and $r \in V_T$.
Let $Q \subseteq V_T$ be a non-empty subset.
We apply the root and prune primitive on $T$.
Let $T_Q = (V_Q, E_Q)$ denote the resulting tree.
Let $\degree_Q(u)$ denote the degree of a node $u \in V_Q$ in $T_Q$.
We define the \emph{augmentation set} to be $A_Q = \{ u \in V_Q \mid \degree_Q(u) \geq 3 \}$.

\begin{lemma}
    Our root and prune primitive computes the augmentation set $A_Q$ within $O(\log |Q|)$ rounds.
\end{lemma}

\begin{proof}
    Note that a neighbor $v \in \neighborhood(u)$ of a node $u \in V_Q$ is in $V_Q$ iff $\prefix_{(u,v)} - \prefix_{(v,u)} \neq 0$.
    This allows each node $u \in V_Q$ to compute $\degree_Q(u)$. 
    Hence, the statement follows from \Cref{lem:rap}.
\end{proof}

Let $Q' = Q \cup A_Q$.
We now prove that $T$ and each connected subgraph of $T$ has one or two $Q'$-centroids, and that $A_Q = O(|Q|)$ and with that $Q' = O(|Q|)$.

\begin{lemma}
\label{lem:centroid:Qprime}
    Let $Q \subseteq V_T$ be a non-empty subset.
    Then, $T$ has one or two $Q'$-centroids.
\end{lemma}

\begin{proof}
    If $|Q| = 1$, then $A_Q = \emptyset$.
    Trivially, the only node in $Q' = Q$ is a $Q'$-centroid.
    Suppose that $|Q| \geq 2$.

    We prune any subtrees with respect to $r$ without a node in $Q'$.
    Then, while the root is not in $Q'$ and has a degree of $1$, we prune the root and make its child the new root.
    Let $T' = (V_{T'},E_{T'})$ denote the resulting tree.
    Note that $|V_{T'}| \geq 2$ since $|Q| \geq 2$.
    This implies that each node has at least a degree of $1$.
    Further, note that each node of degree $1$ in $T'$ is in $Q \subseteq Q'$.
    Otherwise, we could prune the node.

    Consider a node $u \in V_{T'} \setminus Q'$.
    It cannot have a degree of $0$.
    Otherwise, it would have been pruned.
    It cannot have a degree of $1$.
    Otherwise, it would be in $Q'$.
    It cannot have a degree of greater than $2$ since $A_Q \subseteq Q'$.
    Therefore, it must have a degree of $2$.
    These can only form chains with the endpoints adjacent to one node in $Q'$, respectively.
    We replace each of these chains with a single edge.
    Let $T'' = (Q',E_{T''})$ denote the resulting tree.

    We claim that a node is a $Q'$-centroid of $T'$ iff it is a $Q'$-centroid in $T$.
    Further, a node is a $Q'$-centroid of $T''$ iff it is a $Q'$-centroid in $T'$.
    Consider the connected components after the removal of a node $u \in Q'$.
    On the one hand, the removal of nodes not in $Q'$ might eliminate connected components without any nodes in $Q'$.
    Otherwise, it does not affect the number of nodes in $Q'$ of any other connected component.
    On the other hand, the addition of nodes not in $Q'$ might introduce new connected components without any nodes in $Q'$.
    Otherwise, it does not affect the number of nodes in $Q'$ of any existing connected component.
    Hence, both claims hold.
    By \Cref{th:centroid:simple} or \Cref{th:centroid:weighted}, $T''$ has one or two centroids and with that also $T$.
\end{proof}

\begin{corollary}
\label{cor:centroid:Qprime}
    Furthermore, let $Z = (V_Z,E_Z)$ be a connected subgraph of $T$ with $Q' \cap V_Z \neq \emptyset$.
    Then, $Z$ has one or two $(Q' \cap V_Z)$-centroids.
\end{corollary}

\begin{proof}
    First, note that a node is a $(Q' \cap V_Z)$-centroid of $Z$ iff it is a $(Q' \cap V_Z)$-centroid in $T$.
    This holds by the same arguments as in the proof of \Cref{lem:centroid:Qprime}.
    %
    Since $Q' \cap V_Z \neq \emptyset$, $T$ has one or two $((Q' \cap V_Z) \cup A_{Q' \cap V_Z})$-centroids by \Cref{lem:centroid:Qprime}.
    We claim that $A_{Q' \cap V_Z} \subseteq Q' \cap V_Z$ and with that $((Q' \cap V_Z) \cup A_{Q' \cap V_Z}) = Q' \cap V_Z$.
    
    Observe that for each node $u$, $\degree_{Q' \cap V_Z}(u) \leq \degree_{Q'}(u)$ holds.
    This implies $A_{Q' \cap V_Z} \subseteq A_{Q'}$ and with that $A_{Q' \cap V_Z} \subseteq Q'$.
    Let $u \in A_{Q' \cap V_Z}$.
    Since by definition, $\degree_{Q' \cap V_Z}(u) \geq 3$, $u$ connects at least $2$ nodes in $Q' \cap V_Z \subseteq V_Z$ (one edge might lead to $r$).
    Since $Z$ is connected, $u$ has to be in $V_Z$.
    Since $A_{Q' \cap V_Z} \subseteq Q'$ and $A_{Q' \cap V_Z} \subseteq V_Z$, the claim holds.
    By the claim, $T$ has one or two $(Q' \cap V_Z)$-centroids and with that also $Z$.
\end{proof}

\begin{corollary}
\label{cor:centroid:augmentation}
    Let $Q \subseteq V_T$.
    Then, $|A_Q| = O(|Q|)$, and with that $|Q'| = O(|Q|)$.    
\end{corollary}

\begin{proof}
    Consider tree $T'' = (V_{T''},E_{T''})$ of the proof of \Cref{lem:centroid:Qprime}.
    Let $L = \{ u \in V_{T''} \mid \degree_Q(u) = 1 \}$ denote the set of all leaves.
    By definition, no leaf is in $A_Q$.
    By construction, each leaf is in $Q'$.
    By the handshaking lemma and the Euler formula, $|L| + 2(|V_{T''}| - |L|) + 3|A_Q| \leq \sum_{v \in V_{T''}} \degree(v) = 2(V_{T''} - 1)$.
    This implies $|A_Q| \leq |L| - 2 \leq |Q| - 1$.
\end{proof}


\Cref{cor:centroid:Qprime} allows us to decompose a tree $T$ recursively.
In each recursion, we decompose the tree at a $Q'$-centroid and perform a recursion on each resulting subtree if it contains at least one node in $Q'$.
Observe that each node in $Q'$ is chosen at some point to decompose a subtree.
We now define a tree $\decomposition(T) = (Q', E_D)$ as follows.
For each recursion (except the first one), we add an edge from the $Q'$-centroid used in the recursion to the $Q'$-centroid used in the calling recursion.
We call $\decomposition(T)$ a \emph{$Q'$-centroid decomposition tree}.
Note that $T$ may have multiple $Q'$-centroid decomposition trees since in each recursion, we could have two $Q'$-centroids to choose from.

\begin{lemma}
\label{lem:decomposition:height}
    The $Q'$-centroid decomposition tree has height $O(\log |Q|)$.
\end{lemma}

\begin{proof}
    The $Q'$-centroid decomposition tree has height $O(\log |Q'|)$ since each recursion halves the number of nodes in $Q'$.
    By \Cref{cor:centroid:augmentation}, $O(\log |Q'|) = O(\log |Q|)$.
\end{proof}

Let $T = (V_T, E_T)$ be a tree.
Let a node $r \in V_T$ and a subset $Q' \subseteq V_T$ for a non-empty subset $Q \subseteq V_T$ be given, i.e., each node $u \in V_T$ knows whether $u = r$ and whether $u \in Q'$.
The goal of the \emph{$Q'$-centroid decomposition problem} is to iteratively compute a $Q'$-centroid decomposition tree $\decomposition(T) = (Q', E_D)$, i.e., in the $i$-th iteration, each node $u \in V_T$ determines whether it is a node in $\decomposition(T)$ of depth $i-1$.

A recursion of our \emph{decomposition primitive} for a connected subtree $Z = (V_Z, E_Z)$ proceeds as follows.
Let a node $r_Z \in V_Z$ and a subset $Q'_Z \subseteq V_Z$ be given.
For the first recursion, we set $r_Z = r$ and $Q'_Z = Q'$.
First, we apply the centroid primitive with $r_Z$ and $Q'_Z$ to compute the $Q'_Z$-centroids.
Then, we apply the election primitive to elect one of the $Q'_Z$-centroids.
Let $c_Z$ denote the elected centroid.
We decompose the graph at $c_Z$.
Let $Z_u = (V_u, E_u)$ be the subtree containing node $u$.
For each neighbor $u \in \neighborhood(c_Z)$, subtree $Z_u$ establishes a circuit that connects all nodes of the subtree.
Then, each node $v \in Q'_Z$ beeps on the circuit of its subtree.
For each neighbor $u \in \neighborhood(c_Z)$, there is a beep on the circuit of $Z_u$ iff $Q'_Z \cap V_u \neq \emptyset$.
Finally, for each neighbor $u \in \neighborhood(c_Z)$, we perform a recursion on the subgraph $Z_u$ with $r_{Z_u} = u$ and $Q'_{Z_u} = Q'_Z \cap V_u$ if $Q' \cap V_u \neq \emptyset$.

The decomposition primitive performs all recursions of the same recursion level in parallel.
After each execution of a recursion level, the amoebots establish a global circuit.
Each node in $Q'$ that has not been elected so far beeps.
We terminate if there was a beep.
Otherwise, we proceed to the next recursion level.

\begin{lemma}
\label{lem:decomposition}
    Our decomposition primitive computes a \emph{$Q'$-centroid decomposition tree} within $O(\log^2 |Q|)$ rounds.
    More precisely, in the $i$-th recursion level, the nodes compute the $Q'$-centroids of depth $i-1$.
\end{lemma}

\begin{proof}
    \Cref{cor:centroid:Qprime} guarantees that existence of a $Q'$-centroid decomposition tree.
    The correctness follows then from \Cref{lem:centroid,lem:election}.
    
    By \Cref{lem:centroid}, the centroid primitive requires $O(\log |Q|)$ rounds.
    All other steps including the election primitive (see \Cref{lem:election}) only require a constant number of rounds.
    Hence, each recursion requires $O(\log |Q|)$ rounds.
    By \Cref{lem:decomposition:height}, there are $O(\log |Q|)$ recursion levels.
\end{proof}


\subsection{Portal Trees}

Theoretically, we can also apply all primitives in this section on portal graphs (see \Cref{fig:rap:pg}).
However, recall that the amoebot structure has only access to the implicit portal graphs.
In the following, we will first adapt the ETT to implicit portal graphs.
Then, we will redefine and solve all problems.

Let $T = (V_T,E_T)$ be an implicit portal graph of portal graph $\mathcal P = (V_{\mathcal P},E_{\mathcal P})$.
For each $P_1 \in V_{\mathcal P}$ and $P_2 \in \neighborhood(P_1)$, let $\connector_{P_1}(P_2)$ denote the amoebot $u \in P_1$ that is adjacent to an amoebot $v \in P_2$ in $T$, i.e., $\{ u, v \} \in E_T$.
Recall that by construction, $\connector_{P_1}(P_2)$ is unique.

We adapt the ETT to implicit portal graphs as follows.
Let a portal $R \in V_{\mathcal P}$ and a subset $\mathcal Q \subseteq V_{\mathcal P}$ be given, i.e., each amoebot $u \in V_T$ knows whether $\portal(u) = R$ and whether $\portal(u) \in \mathcal Q$.
For each portal, we elect a representative, e.g., for each $x$-portal, we elect the westernmost amoebot.
Let $\hat r \in V_T$ be the representative of the root portal $R$.
Let $\hat Q \subseteq V_T$ be the set of representatives of the portals in $\mathcal Q$.
Consider an execution of the ETT on the implicit portal graph with weight function $\weight_{\hat Q}$ (see \Cref{fig:rap:ipg}).
Observe that by \Cref{cor:euler:number}, amoebot $\hat r$ computes $|\hat Q| = |\mathcal Q|$.
We now compare that execution to an execution of the ETT on the portal graph with weight function $\weight_{\mathcal Q}$ (see \Cref{fig:rap:pg}).

\begin{lemma}
\label{lem:rap:comparison}
    Let $P_1, P_2 \in V_{\mathcal P}$ be two adjacent portals.
    Let $u_1 = \connector_{P_1}(P_2)$ and $u_2 = \connector_{P_2}(P_1)$.
    Then, $\prefix_{(P_2,P_1)} - \prefix_{(P_1,P_2)} = \prefix_{(u_2,u_1)} - \prefix_{(u_1,u_2)}$ holds.
\end{lemma}

\iftoggle{full}{
\begin{proof}
    W.l.o.g., let $P_1$ be the parent of $P_2$ with respect to $R$.
    Observe that the subtree of $u_2$ with respect to $\hat r$ in the implicit portal graph contains the same amoebots as the union of the portals in the subtree of $P_2$ with respect to $R$ in the portal graph.
    The equation holds since by construction, we mark the same number of edges on the subpaths of the Euler tours, respectively.
\end{proof}
}{}

Thus, the ETT computes all the necessary information, i.e., $|\mathcal Q|$ and all prefix sum differences, to solve all problems.
However, this information is scattered across the amoebots.
More precisely, in portal $P_1 \in V_{\mathcal P}$, for each $P_2 \in \neighborhood(P_1)$, only $\connector_{P_1}(P_2)$ computes $\prefix_{(P_1,P_2)} - \prefix_{(P_2,P_1)}$ and with that all further computations and comparisons.
Hence, the amoebots of each portal have to communicate with each other.
In the following, we explain how the necessary communication for each problem.


For the \emph{root and prune problem}, let a portal $R \in V_{\mathcal P}$ and a subset $\mathcal Q \subseteq V_{\mathcal P}$ be given, i.e., each amoebot $u \in V_T$ knows whether $\portal(u) = R$ and whether $\portal(u) \in \mathcal Q$.
Let $\subtree_R(P)$ denote the set of portals in the subtree of $P$ in the tree rooted at $R$.
Let $V_{\mathcal Q} = \{P \in V_{\mathcal P} \mid \subtree_R(u) \cap Q \neq \emptyset \}$ denote the set of portals whose subtree contains portals in $\mathcal Q$.
Our goal is to root $\mathcal P$ at $R$ and to prune all subtrees without a node in $\mathcal Q$, i.e., (i) each node $u \in V_T$ determines whether $\portal(u) \in V_{\mathcal Q}$, and (ii) each amoebot $u \in V_T$ where $\portal(u) \in V_{\mathcal Q} \setminus \{ R \}$ identifies all neighbors $v \in \neighborhood(u)$ such that $\portal(v)$ is the parent of $\portal(u)$ with respect to $R$.

We apply the ETT on $T$ as described above.
For each $P_1 \in V_{\mathcal P} \setminus \{ R \}$ and each $P_2 \in \neighborhood(P_1)$, amoebot $\connector_{P_1}(P_2)$ computes $\prefix_{(P_1,P_2)} - \prefix_{(P_2,P_1)} \neq 0$ and compares it with $0$.
At the same time, $\hat r$ computes $|\mathcal Q|$.

In order to determine whether $\portal(u) \in V_{\mathcal Q}$ for each $u \in V_T$, we proceed as follows.
First, each portal establishes a circuit that connects all amoebots of the portal (see \Cref{fig:rap:circuits:1}).
Then, for each $P_1 \in V_{\mathcal P} \setminus \{ R \}$ and each $P_2 \in \neighborhood(P_1)$, amoebot $\connector_{P_1}(P_2)$ beeps on the circuit of $P_1$ if $\prefix_{(P_1,P_2)} - \prefix_{(P_2,P_1)} \neq 0$.
Simultaneously, amoebot $\hat r \in R$ beeps on the circuit of $R$ if $|\mathcal Q| > 0$.
For each amoebot $u \in V_T$, $\portal(u) \in V_{\mathcal Q}$ if there was a beep on the circuit of $\portal(u)$.

\begin{figure}[tbp]
    \centering
    \begin{minipage}[t]{.49\linewidth}
        \centering
        \includegraphics{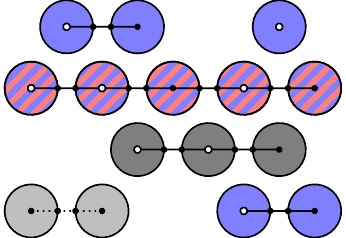}
        \subcaption{Circuits for each portal.}
        \label{fig:rap:circuits:1}
    \end{minipage}
    \hfill
    \begin{minipage}[t]{.49\linewidth}
        \centering
        \includegraphics{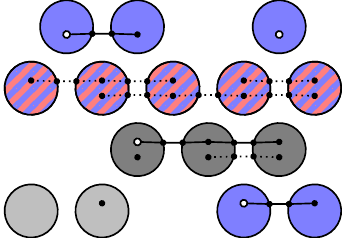}
        \subcaption{Circuits for each directed edge.}
        \label{fig:rap:circuits:2}
    \end{minipage}
    \caption{
        Used circuits for the root and prune primitive.
        For the meaning of the colors, see \Cref{fig:rap}.
        The white partition sets and the solid circuits indicate beeps.
    }
    \label{fig:rap:circuits}
\end{figure}

In order to identify all neighbors $v \in \neighborhood(u)$ such that $\portal(v)$ is the parent of $\portal(u)$ with respect to $R$ for each $u \in V_T$, we proceed as follows.
First, each portal $P_1$ establishes a circuit for each adjacent portal $P_2$ that connects all amoebots in $P_1$ that are adjacent to an amoebot in $P_2$ (see \Cref{fig:rap:circuits:2}).
For the sake of simplicity, we attribute that circuits to the directed edge $(P_1,P_2)$, respectively.
Then, for each $P_1 \in V_{\mathcal P} \setminus \{ R \}$ and each $P_2 \in \neighborhood(P_1)$, amoebot $\connector_{P_1}(P_2)$ beeps on the circuit of $(P_1,P_2)$ if $\prefix_{(P_1,P_2)} - \prefix_{(P_2,P_1)} > 0$.
A beep on the circuit of $(P_1,P_2)$ indicates that $P_2$ is the parent of $P_1$ with respect to $R$.
Hence, each $u \in P_1$ is able to identify its neighbors in $P_2$.

\begin{lemma}
\label{lem:portal:rap}
    Our \emph{root and prune primitive} roots $\mathcal P$ at $R$ and prunes all subtrees without a portal in $\mathcal Q$ within $O(\log |\mathcal Q|)$ rounds.
\end{lemma}

In order to compute the augmentation set $A_{\mathcal Q}$ of $\mathcal Q$, we have to compute the degree of each portal in $V_{\mathcal Q}$.
For that, each portal $P_1 \in V_{\mathcal P}$ simply applies the PASC algorithm where for each $P_2 \in \neighborhood(P_1)$, amoebot $\connector_{P_1}(P_2)$ participates if $\prefix_{(P_1,P_2)} - \prefix_{(P_2,P_1)} \neq 0$.
In case that there is an amoebot $u = \connector_{P_1}(P_2) = \connector_{P_1}(P_3)$ for $P_2 \neq P_3$ with $\prefix_{(P_1,P_2)} - \prefix_{(P_2,P_1)} \neq 0$ and $\prefix_{(P_1,P_3)} - \prefix_{(P_3,P_1)} \neq 0$, then $u$ simulates two participating amoebots in the chain.
By \Cref{cor:pasc:prefixsum}, the last amoebot of the portal $P_1$ computes $\degree_{\mathcal Q}(P_1)$.
It compares the degree with $3$.
Now, each portal establishes a circuit that connects all amoebots of the portal (compare to \Cref{fig:rap:circuits:1}).
Then, the last amoebot of portal $P_1$ beeps on the circuit of $P_1$ if $\degree_{\mathcal Q}(P_1) \geq 3$.
For each amoebot $u \in V_T$, $\portal(u) \in A_{\mathcal Q}$ if there was a beep on the circuit of $\portal(u)$.

\begin{lemma}
\label{lem:portal:augmentation}
    Our \emph{root and prune primitive} computes the augmentation set $A_\mathcal Q$ within $O(\log |\mathcal Q|)$ rounds.
\end{lemma}


For the \emph{election problem}, let a portal $R \in V_{\mathcal P}$ and a subset $\mathcal Q \subseteq V_{\mathcal P}$ be given, i.e., each amoebot $u \in V_T$ knows whether $\portal(u) = R$ and whether $\portal(u) \in \mathcal Q$.
Our goal is to elect a single portal $R' \in \mathcal Q$, i.e., each node $u \in V_T$ knows whether $\portal(u) = R'$.

We apply the simplified ETT on $T$ to elect an amoebot $r' \in \hat Q$.
Then, we elect $R' = \portal(r') \in \mathcal Q$.
For that, each portal establishes a circuit that connects all amoebots of the portal (compare to \Cref{fig:rap:circuits:1}).
Amoebot $r'$ beeps on the circuit of its portal.
Each amoebot that receives a beep belongs to $R' = \portal(r')$.

\begin{lemma}
\label{lem:portal:election}
    Our \emph{election primitive} elects a single portal $R' \in \mathcal Q$ within $O(1)$ rounds.
\end{lemma}


For the \emph{$\mathcal Q$-centroid problem}, let a portal $R \in V_{\mathcal P}$ and a subset $\mathcal Q \subseteq V_{\mathcal P}$ be given, i.e., each amoebot $u \in V_T$ knows whether $\portal(u) = R$ and whether $\portal(u) \in \mathcal Q$.
Our goal is to compute the $\mathcal Q$-centroid(s) of $\mathcal P$, i.e., each node $u \in V_T$ determines whether $\portal(u)$ is a $\mathcal Q$-centroid.

First, we apply the root and prune primitive to compute the parent of each portal in $\mathcal Q$ with respect to $R$.
Then, we apply the ETT on $T$ as described above.
For each $P_1 \in V_{\mathcal P} \setminus \{ R \}$ and each $P_2 \in \neighborhood(P_1)$, amoebot $\connector_{P_1}(P_2)$ computes $\size_{P_1}(P_2)$ and compares it to $|\mathcal Q|/2$ where $\size_{P_1}(P_2) = |\mathcal Q| - \prefix_{(P_1,P_2)} - \prefix_{(P_2,P_1)}$ if $P_2$ is the parent of $P_1$ and $\size_{P_1}(P_2) = \prefix_{(P_2,P_1)} - \prefix_{(P_1,P_2)}$ otherwise.
For that, $\hat r$ computes and broadcasts $|\mathcal Q|$.

Then, each portal establishes a circuit that connects all amoebots of the portal (compare to \Cref{fig:rap:circuits:1}).
For each $P_1 \in V_{\mathcal P} \setminus \{ R \}$ and each $P_2 \in \neighborhood(P_1)$, amoebot $\connector_{P_1}(P_2)$ beeps on the circuit of $P_1$ if $\size_{P_1} > |\mathcal Q|/2$.
A beep on the circuit of $P_1$ indicates that $P_1$ cannot be a $\mathcal Q$-centroid.
Hence, each amoebot $u \in V_T$ where $\portal(u) \in V_{\mathcal Q}$ that does not receive a beep belongs to a $\mathcal Q$-centroid.

\begin{lemma}
\label{lem:portal:centroid}
    Our \emph{$\mathcal Q$-centroid primitive} computes the $\mathcal Q$-centroids within $O(\log |\mathcal Q|)$ rounds.
\end{lemma}


For the \emph{$\mathcal Q'$-centroid decomposition problem}, let a portal $R \in V_{\mathcal P}$ and a subset $\mathcal Q' \subseteq V_{\mathcal P}$ for a non-empty subset $Q \subseteq V_{\mathcal P}$ be given, i.e., each amoebot $u \in V_T$ knows whether $\portal(u) = R$ and whether $\portal(u) \in \mathcal Q'$.
Our goal is to iteratively compute a $\mathcal Q'$-centroid decomposition tree $\decomposition(\mathcal P) = (\mathcal Q', E_D)$, i.e., in the $i$-th iteration, each node $u \in V_T$ determines whether $\portal(u)$ is a portal in $\decomposition(\mathcal P)$ of depth $i-1$.

Our \emph{decomposition primitive} for implicit portal graphs works largly like the decomposition primitive for general tree structure, i.e., a recursion for a connected subtree $\mathcal Z = (V_\mathcal Z, E_\mathcal Z)$ of $\mathcal P$ proceeds as follows.
Let a node $R_\mathcal Z \in V_\mathcal Z$ and a subset $\mathcal Q'_\mathcal Z \subseteq V_\mathcal Z$ be given.
For the first recursion, we set $R_\mathcal Z = R$ and $\mathcal Q'_\mathcal Z = \mathcal Q'$.
First, we apply the centroid primitive with $R_\mathcal Z$ and $\mathcal Q'_\mathcal Z$ to compute the $\mathcal Q'_\mathcal Z$-centroids.
Then, we apply the election primitive to elect one of the $\mathcal Q'_\mathcal Z$-centroids.
Let $C_\mathcal Z$ denote the elected centroid.
We decompose the subtree $\mathcal Z$ at $C_\mathcal Z$.

Let $\mathcal Z_P = (V_P, E_P)$ denote the subtree containing portal $P \in \neighborhood(C_\mathcal Z)$.
Let $R_{\mathcal Z_P} = P$ and $\mathcal Q'_{\mathcal Z_P} = \mathcal Q'_\mathcal Z \cap V_P$.
For each of these subtrees $\mathcal Z_P$, we proceed as follows.
First, each portal establishes a circuit that connects all amoebots of the portal.
Then, $\connector_P(C_\mathcal Z)$ beeps on the circuit of $P$.
Each amoebot that receives a beep belongs to $R_{\mathcal Z_P}$.
Second, the subtree establishes a circuit that connects all amoebots of the subtree.
Then, each amoebot of a portal in $\mathcal Q'_\mathcal Z$ beeps on the circuit of the subtree.
There is a beep on the circuit of the subtree iff $\mathcal Q'_\mathcal Z \cap V_P \neq \emptyset$.
Finally, we perform a recursion on the subtree $\mathcal Z_P$ with $R_{\mathcal Z_P}$ and $\mathcal Q'_{\mathcal Z_P}$ if $\mathcal Q' \cap V_P \neq \emptyset$.

The decomposition primitive performs all recursions of the same recursion level in parallel.
After each execution of a recursion level, the amoebots establish a global circuit.
Each amoebot of a portal in $\mathcal Q'$ that has not been elected so far beeps.
We terminate if there was a beep.
Otherwise, we proceed to the next recursion level.

\begin{lemma}
\label{lem:portal:decomposition}
    Our decomposition primitive computes a \emph{$\mathcal Q'$-centroid decomposition tree} within $O(\log^2 |Q|)$ rounds.
    More precisely, in the $i$-th recursion level, the nodes compute the $\mathcal Q'$-centroids of depth $i-1$.
\end{lemma}

}{


\section{Tree Primitives}

Our algorithms make excessive use of tree structures.
Due to space constraints, we will only list the tree primitives that we use.
Let $T = (V_T,E_T)$ be an implicit portal graph of portal graph $\mathcal P = (V_{\mathcal P},E_{\mathcal P})$.
Let a portal $R \in V_{\mathcal P}$ and a subset $\mathcal Q \subseteq V_{\mathcal P}$ be given, i.e., each amoebot $u \in V_T$ knows whether $\portal(u) = R$ and whether $\portal(u) \in \mathcal Q$.


Let $\subtree_R(P)$ denote the set of portals in the subtree of $P$ in the tree rooted at $R$.
Let $V_{\mathcal Q} = \{P \in V_{\mathcal P} \mid \subtree_R(u) \cap Q \neq \emptyset \}$ denote the set of portals whose subtree contains portals in $\mathcal Q$.

\begin{lemma}
\label{lem:portal:rap}
    Our \emph{root and prune primitive} roots $\mathcal P$ at $R$ and prunes all subtrees without a portal in $\mathcal Q$ within $O(\log |\mathcal Q|)$ rounds.
    More precisely, each node $u \in V_T$ determines whether $\portal(u) \in V_{\mathcal Q}$, and each amoebot $u \in V_T$ where $\portal(u) \in V_{\mathcal Q} \setminus \{ R \}$ identifies all neighbors $v \in \neighborhood(u)$ such that $\portal(v)$ is the parent of $\portal(u)$ with respect to $R$.
\end{lemma}


Let $T_\mathcal Q = (V_\mathcal Q, E_\mathcal Q)$ denote the resulting tree.
Let $\degree_\mathcal Q(u)$ denote the degree of a portal $P \in V_\mathcal Q$ in $T_\mathcal Q$.
We define the \emph{augmentation set} to be $A_Q = \{ u \in V_Q \mid \degree_Q(u) \geq 3 \}$.

\begin{lemma}
\label{lem:portal:augmentation}
    Our \emph{augmentation primitive} computes the augmentation set $A_\mathcal Q$ within $O(\log |\mathcal Q|)$ rounds.
\end{lemma}


\begin{lemma}
\label{lem:portal:election}
    Our \emph{election primitive} elects a single portal $R' \in \mathcal Q$ within $O(1)$ rounds.
    More precisely, each node $u \in V_T$ knows whether $\portal(u) = R'$.
\end{lemma}


A portal $P \in \mathcal Q$ is a \emph{$\mathcal Q$-centroid} iff after the removal of $P$, the tree splits into connected components with at most $|\mathcal Q|/2$ portals in $\mathcal Q$, respectively.

\begin{lemma}
\label{lem:portal:centroid}
    Our \emph{$\mathcal Q$-centroid primitive} computes the $\mathcal Q$-centroids within $O(\log |\mathcal Q|)$ rounds.
    More precisely, each portal $u \in V_T$ determines whether $\portal(u)$ is a $\mathcal Q$-centroid.
\end{lemma}


Now, let additionally $A_\mathcal Q \subseteq V_{\mathcal P}$ be given, i.e., each amoebot $u \in V_T$ knows whether $\portal(u) \in A_\mathcal Q$.
For $\mathcal Q' = \mathcal Q \cup A_\mathcal Q$, we can show that each connected subtree of $\mathcal P$ has one or two $\mathcal Q'$-centroids.
This allows us to decompose $\mathcal P$ recursively.
In each recursion, we decompose the tree at a $\mathcal Q'$-centroid and perform a recursion on each resulting subtree if it contains at least one node in $\mathcal Q'$.
Observe that each node in $\mathcal Q'$ is chosen at some point to decompose a subtree.
We now define a tree $\decomposition(\mathcal P) = (\mathcal Q', E_D)$ as follows.
For each recursion (except the first one), we add an edge from the $\mathcal Q'$-centroid used in the recursion to the $\mathcal Q'$-centroid used in the calling recursion.
We call $\decomposition(\mathcal P)$ a \emph{$\mathcal Q'$-centroid decomposition tree}.

\begin{lemma}
\label{lem:portal:decomposition}
    Our \emph{decomposition primitive} computes a \emph{$\mathcal Q'$-centroid decomposition tree} within $O(\log^2 |\mathcal Q|)$ rounds.
    More precisely, in the $i$-th recursion level, we compute the $\mathcal Q'$-centroids of depth $i-1$, i.e., each amoebot $u \in V_T$ determines whether $\portal(u)$ is a $\mathcal Q'$-centroid of depth $i-1$.
\end{lemma}

}



\section{Shortest Path Forest with a Single Source}
\label{sec:spt}

In this section, we consider $(1,\ell)$-SPF.
Let $s$ denote the only amoebot in $S$.
An amoebot $v \in \neighborhood(u)$ is a feasible parent of $u \in X$ with respect to $s$ iff $\dist(s,u) = \dist(s,v) + 1$, which is equivalent to
\begin{equation}
\label{eq:sssp:parent}
    \Big(\dist_x(s,u) - \dist_x(s,v)\Big) + \Big(\dist_y(s,u) - \dist_y(s,v)\Big) + \Big(\dist_z(s,u) - \dist_z(s,v)\Big) = 2
\end{equation}
by \Cref{lem:portal_graph:triangular}.
Since each portal graph is a tree (see \Cref{lem:portal_graph:tree}), each difference $\dist_d(s,u) - \dist_d(s,v)$ only depends on the relation between $\portal_d(u)$ and $\portal_d(v)$.
If $\portal_d(v)$ is the parent of $\portal_d(u)$, then $\dist_d(s,u) - \dist_d(s,v) = 1$.
If $\portal_d(v)$ is equal to $\portal_d(u)$, then $\dist_d(s,u) - \dist_d(s,v) = 0$.
If $\portal_d(v)$ is a child of $\portal_d(u)$, then $\dist_d(s,u) - \dist_d(s,v) = -1$.
Note that since $v \in \neighborhood(u)$, there must be an axis $d$ such that $\portal_d(v)$ is equal to $\portal_d(u)$.
This implies that \Cref{eq:sssp:parent} only holds if $\portal_d(v)$ is the parent of $\portal_d(u)$ for the other two axes.
%
Hence, computing all relative distances reduces to rooting the portal graphs at $s$.
For that, we apply the root and prune primitive on each portal graph with $s$ and $D$ as parameters.

\begin{lemma}
\label{lem:sssp:parent:shortest_path}
    Each amoebot on the shortest path from $s$ to $D$ is able to choose a parent with respect to $s$.
\end{lemma}

\iftoggle{full}{
\begin{proof}
    Suppose the contrary, i.e., there is an amoebot $u$ on the shortest path from $s$ to $v \in D$ that is not able to choose a parent with respect to \Cref{eq:sssp:parent}.
    This can only happen if the root and prune primitive has pruned $\portal(u)$ for at least one of the (implicit) portal graphs.
    This implies that the subtree of $\portal(u)$ does not contain any portals with amoebots in $D$.
    But then, the shortest path cannot traverse any amoebot in $\portal(u)$ (including $u$).
    Otherwise, we could shorten the shortest path.
    This is a contradiction to the assumption.
\end{proof}
}{}

\Cref{lem:sssp:parent:shortest_path} guarantees that we obtain a shortest path from $s$ to each $u \in D$.
However, it does not prevent other amoebots to choose a parent with respect to $s$ (see \Cref{fig:spt}).
As a result, we might obtain a shortest path tree with subtrees without amoebots in $D$, and additional connected components (without amoebots in $D$).
In order to remove such subtrees and connected components, we simply apply the root and prune algorithm on the resulting graph with $s$ and $D$ as parameters.
The algorithm prunes the subtrees by definition.
Furthermore, during the execution, the connected components that do not contain $s$ do not receive any signals.
This allows us to also prune these connected components.

\begin{figure}[tbp]
    \begin{minipage}[t]{.27\linewidth}
        \centering
        \includegraphics{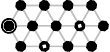}
        \subcaption{Initial structure.}
        \label{fig:spt:full}
    \end{minipage}
    \hfill
    \begin{minipage}[t]{.37\linewidth}
        \centering
        \includegraphics{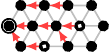}
        \subcaption{Parents w.r.t. the portal graphs.}
        \label{fig:spt:result}
    \end{minipage}
    \hfill
    \begin{minipage}[t]{.32\linewidth}
        \centering
        \includegraphics{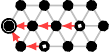}
        \subcaption{Parents after the final pruning.}
        \label{fig:spt:final_pruning}
    \end{minipage}
    \caption{
        Shortest path tree algorithm.
        The encircled amoebot indicates $s$.
        The amoebots marked by a white dot indicate $D$.
        The red arrows indicate the parents.
    }
    \label{fig:spt}
\end{figure}

\begin{theorem}
\label{th:spt}
    The shortest path tree algorithm computes an $(\{ s \},D)$-shortest path forest within $O(\log \ell)$ rounds.
\end{theorem}

\iftoggle{full}{
\begin{proof}
    By \Cref{lem:sssp:parent:shortest_path}, we obtain a shortest path forest that contains a shortest path tree rooted at $s$ that contains all amoebots in $D$.
    By \Cref{lem:portal:rap}, an additional execution of the root and prune primitive extracts that tree.
    Overall, the algorithm consists of $4$ executions of the root and prune algorithm.
    Each execution requires $O(\log \ell)$ rounds.
\end{proof}
}{}



\section{Shortest Path Forests with Multiple Sources}
\label{sec:spf}

In this section, we consider $(k,\ell)$-SPF for arbitrary $k$.
In the first subsection, we will start with the simple case where the amoebot structure forms a line for which we will present an $O(\log n)$ algorithm.


In the second subsection, we will show how to merge two shortest path forests into a single one within $O(\log n)$ rounds.
This leads immediately to the following naive solution.
Suppose that we have already computed an $S'$-shortest path forest for a subset $S' \subseteq S$ of the sources.
Now, compute an $\{s\}$-shortest path forest for an arbitrary source $s \in S \setminus S'$.
Then, merge both shortest path forests to an $S' \cup \{s\}$-shortest path forest.
This sequential approach requires $O(k \log n)$ rounds.


We can improve the runtime by applying a divide and conquer approach.
The idea of our shortest path forest algorithm is to split the amoebot structure at a portal, to recursively compute a shortest path forest for both sides, and to finally merge them.
The first three subsections contain elementary procedures used in our algorithm.
In the fourth subsection, we will elaborate on our divide and conquer approach.


\subsection{Line Algorithm}
\label{sec:line}

Suppose the amoebot structure $X$ forms a line.
The \emph{line algorithm} computes an $S$-shortest path forest as follows.
Observe that the closest source of each amoebot $u \in X \setminus S$ must be the next source in one of both directions.
Hence, it suffices if it computes its distance to those two sources.
For that, we apply the PASC algorithm (see \Cref{lem:pasc:chain}) from each source into both directions up to the next source, respectively (see \Cref{fig:spf:line}).

\begin{lemma}
\label{lem:spf:line}
    The line algorithm computes an $S$-shortest path forest for a line of amoebots within $O(\log n)$ rounds.
\end{lemma}

\iftoggle{full}{
\begin{proof}
    The correctness follows from our observation.
    The algorithm requires $O(\log n)$ rounds since all $2k$ applications of the PASC algorithm are performed in parallel.
\end{proof}
}{}

\begin{figure}[tbp]
    \centering
    \includegraphics{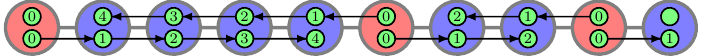}
    \iftoggle{full}{
    
    \caption{
        Line algorithm.
        The red amoebots indicate $S$.
        Note that the easternmost amoebot only receives a distance of one source.
    }
    
    }{
    
    \caption{
        Line algorithm.
        The red amoebots indicate $S$.
    }
    
    }
    \label{fig:spf:line}
\end{figure}


\subsection{Merging Algorithm}
\label{sec:merging}

Let an $S_1$-shortest path forest and an $S_2$-shortest path forest be given.
The \emph{merging algorithm} computes an $(S_1 \cup S_2)$-shortest path forest as follows.
It makes use of the following lemma.

\begin{lemma}
\label{lem:spf:merging}
    Let an $S_1$-shortest path forest and an $S_2$-shortest path forest be given.
    Let $u \in X \setminus (S_1 \cup S_2)$.
    Let $p_1$ denote the parent of $u$ in the $S_1$-shortest path forest, and $p_2$ the parent of $u$ in the $S_2$-shortest path forest.
    Then, $p_1$ is a feasible parent of $u$ in an $(S_1 \cup S_2)$-shortest path forest if $\dist(S_1, u) \leq \dist(S_2, u)$, and $p_2$ is a feasible parent of $u$ in an $(S_1 \cup S_2)$-shortest path forest if $\dist(S_2, u) \leq \dist(S_1, u)$.
\end{lemma}

\iftoggle{full}{
\begin{proof}
    Note that $\dist(S_1 \cup S_2, v) = \min \{ \dist(S_1, v), \dist(S_2, v) \}$ for all $v \in X$.
    Further, note that by definition, $\dist(S_1, p_1) + 1 = \dist(S_1, u)$ and $\dist(S_2, p_2) + 1 = \dist(S_2, u)$.
    Suppose that $\dist(S_1, u) \leq \dist(S_2, u)$.
    This also implies $\dist(S_1, p_1) \leq \dist(S_2, p_1)$ since $\dist(S_2, v) + 1 \geq \dist(S_2, u)$ for all $v \in \neighborhood(u)$.
    Altogether, we obtain
    \begin{align*}
        \dist(S_1 \cup S_2, u)
        &= \min \{ \dist(S_1, u), \dist(S_2, u) \} \\
        &= \dist(S_1, u) \\
        &= \dist(S_1, p_1) + 1 \\
        &= \min \{ \dist(S_1, p_1), \dist(S_2, p_1) \} + 1 \\
        &= \dist(S_1 \cup S_2, p_1) + 1.
    \end{align*}
    The equations imply that $p_1$ is a feasible parent of $u$ in an $(S_1 \cup S_2)$-shortest path forest.
    The second case where $\dist(S_2, u) \leq \dist(S_1, u)$ can be proven analogously.
\end{proof}
}{}

We apply \Cref{lem:spf:merging} to compute an $(S_1 \cup S_2)$-shortest path forest as follows.
In order to compute $\dist(S_1, v)$ and $\dist(S_2, v)$ for each $v \in V$, we apply the PASC algorithm on the $S_1$- and $S_2$-shortest path forest, respectively (see \Cref{cor:pasc:tree}).
This allows each amoebot $u \in X \setminus (S_1 \cup S_2)$ to compare $\dist(S_1, u)$ and $\dist(S_2, u)$ to determine a feasible parent.

\begin{lemma}
\label{lem:spf:merging:algorithm}
    Let an $S_1$-shortest path forest and an $S_2$-shortest path forest be given.
    The merging algorithm computes an $(S_1 \cup S_2)$-shortest path forest within $O(\log n)$ rounds.
\end{lemma}

\iftoggle{full}{
\begin{proof}
    The correctness follows from \Cref{lem:spf:merging}.
    The algorithm requires $O(\log n)$ rounds since we only apply the PASC algorithm.
\end{proof}
}{}


\subsection{Propagation Algorithm}
\label{sec:propagation}

Consider an arbitrary portal $P$ (see \Cref{fig:spf:propagation}).
W.l.o.g., we assume that $P$ is an $x$-portal.
The portal divides the amoebot structure into two sides $A$ and $B$.
Note that each side might be empty.
Let $S \subseteq A \cup P$.
Suppose that we have a $S$-shortest path forest for $A \cup P$.
The \emph{propagation algorithm} propagates the $S$-shortest path forest into $B$, i.e, it computes a $S$-shortest path forest for $A \cup P \cup B$, as follows.

\begin{figure}[tbp]
    \begin{minipage}[t]{.48\linewidth}
        \centering
        \includegraphics{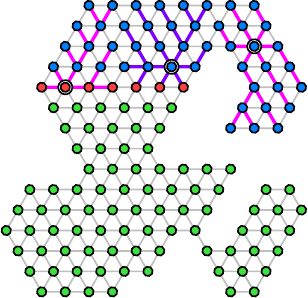}
    \end{minipage}
    \hfill
    \begin{minipage}[t]{.48\linewidth}
        \centering
        \includegraphics{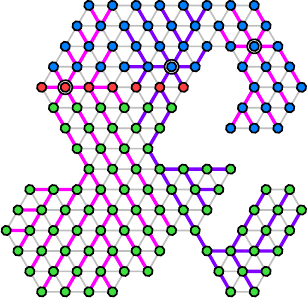}
    \end{minipage}
    \caption{
        Propagation algorithm.
        The left figure shows the initial situation, and the right figure the result.
        The red amoebots indicate $P$.
        The blue and green amoebots indicate $A$ and $B$, respectively.
        The encircled amoebots indicate $S$.
        The pink and purple edges indicate the shortest path trees.
    }
    \label{fig:spf:propagation}
\end{figure}

We say that an amoebot $v \in X$ is \emph{visible} by $u \in X$ iff $v \in \portal_x(u) \cup \portal_y(u) \cup \portal_z(u)$.
Let $\vis(P) = \bigcup_{u \in P}(\portal_x(u) \cup \portal_y(u) \cup \portal_z(u))$ be the \emph{visibility region} of $P$ (see \Cref{fig:spf:propagation:vis}).
The algorithm consists of two phases.
In the first phase, we propagate the shortest path forest into $B' = B \cap \vis(P)$.
In the second phase, we propagate the shortest path forest into $B'' = B \setminus \vis(P)$.
In the following, we will explain both phases.

\begin{figure}[tbp]
    \begin{minipage}[t]{.48\linewidth}
        \centering
        \includegraphics{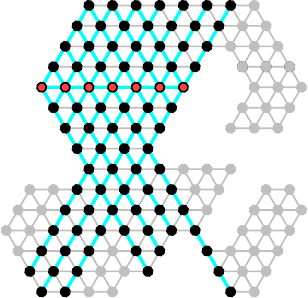}
    \end{minipage}
    \hfill
    \begin{minipage}[t]{.48\linewidth}
        \centering
        \includegraphics{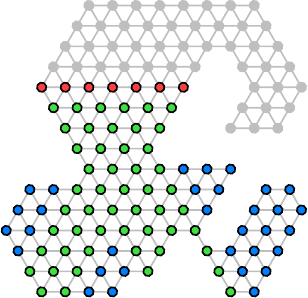}
    \end{minipage}
    \caption{
        Visibility region.
        The red amoebots indicate $P$.
        The cyan edges indicate the portals of each $u \in P$.
        The visibility region $\vis(P)$ contains the red and black amoebots.
        The green amoebots indicate $B'$.
        The blue amoebots indicate $B''$.
    }
    \label{fig:spf:propagation:vis}
\end{figure}


\iftoggle{full}{


Consider the \textbf{first phase}, i.e., the propagation into $B' = B \cap \vis(P)$.
Let $\projy{u}$, $\projz{u} \in V_\Delta$ be the projections of $u \in B'$ along the $y$- and $z$-axis onto the axis of $P$, respectively (see \Cref{fig:spf:triangle}).
Let $\Delta_u$ be the set of nodes in $V_\Delta$ within the triangle defined by $u$, $\projy{u}$, and $\projz{u}$.
Let $P_u = \Delta_u \cap P$.
Further, let $P^W_u$ denote the amoebots in $P$ to the west of $\projz{u}$, and $P^E_u$ denote the amoebots in $P$ to the east of $\projy{u}$.
Note that these might be empty.

\begin{figure}[tbp]
    \begin{minipage}[t]{.49\linewidth}
        \centering
        \includegraphics{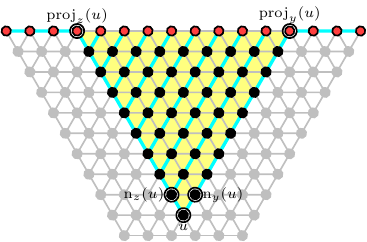}
        \subcaption{Case $\Delta_u \subseteq R$.}
        \label{fig:spf:triangle:case_subset}
    \end{minipage}
    \hfill
    \begin{minipage}[t]{.49\linewidth}
        \centering
        \includegraphics{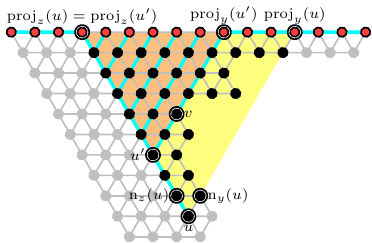}
        \subcaption{Case $\Delta_u \not\subseteq R$.}
        \label{fig:spf:triangle:case_not_subset}
    \end{minipage}
    \caption{
        $(\{u\},P)$-shortest path forest for $P \cup \Delta_u$.
        The labels refer to the encircled amoebots.
    }
    \label{fig:spf:triangle}
\end{figure}

\begin{lemma}
\label{lem:spf:triangle}
    For each $u \in B'$ and $s \in S$, there is a shortest path from $s$ to $u$ within $A \cup P \cup \Delta_u$.
\end{lemma}

\begin{proof}
    Note that there must be a shortest path from $s$ to $u$ since the amoebot structure is connected.
    By \Cref{obs:portal}, each shortest path has to traverse portal $P$.
    Hence, it suffices to show that for each $p \in P$, there is a shortest path from $s$ and $u$ to $p$ within $A \cup P \cup \Delta_u$, respectively.
    
    First, consider a shortest path from $s$ to $p$.
    Due to \Cref{obs:portal}, it cannot traverse any $x$-portal south of $P$.
    Therefore, each shortest path from $s$ to $p$ must be within $A \cup P$.

    Next, we construct a shortest path from $u$ to $p$ as follows.
    Observe that a path between two amoebots $v,w$ is a shortest path if it follows the boundary of the parallelogram spanned by $v$ and $w$.
    For $\Delta_u \subseteq X$, consider \Cref{fig:spf:triangle:case_subset}.
    The path from $u$ to each $p \in P$ satisfies our observation.
    %
    For $\Delta_u \not\subseteq X$, consider \Cref{fig:spf:triangle:case_not_subset}.
    W.l.o.g, let $u$ be visible from $\projz{u}$, i.e., $u \in \portal_z(\projz{u})$.
    Let $u'$ be the amoebot on $\portal_z(u)$ closest to $u$ such that $\Delta_{u'} \subseteq X$.
    The path from $u$ to each $p \in P^W_{u'} \cup P_{u'}$ satisfies our observation.
    %
    Let $p \in P^E_{u'}$.
    Due to \Cref{obs:portal}, a shortest path from $u$ to $p$ cannot traverse any $y$-portal west of $\portal_y(u')$.
    Let $v$ be a boundary amoebot on $\portal_y(u')$ between $u'$ and $\projy{u'}$.
    By \Cref{obs:portal}, each shortest path has to traverse $v$.
    Note that $v$ exists since otherwise, there would be an amoebot $u''$ on $\portal_z(u)$ closer to $u$ than $u'$ such that $\Delta_{u''} \subseteq X$.
    Both, the subpath from $u$ to $v$ and from $v$ to $p$ satisfy our observation.
    Overall, for each $p \in P$, there is a shortest path from $u$ to $p$ within $P \cup \Delta_u$.
\end{proof}

\begin{corollary}
\label{cor:spf:triangle:portal}
    For each $u \in B'$ and $s \in S$, there is a shortest path from $s$ to $u$ that traverses $P_u$.
\end{corollary}

Let $\neighbory{u},\neighborz{u} \in V_\Delta$ denote $u$'s neighbor into $\projy{u}$'s and $\projz{u}$'s direction, respectively (see \Cref{fig:spf:triangle}).
Since $u \in \vis(P)$, at least one of $\neighbory{u}$ and $\neighborz{u}$ has to be in $X$.

\begin{corollary}
\label{cor:spf:triangle:neighbors}
    For each $u \in B'$ and $s \in S$, there is a shortest path from $s$ to $u$ that traverses either $\neighbory{u}$ or $\neighborz{u}$, i.e., either $\neighbory{u}$ or $\neighborz{u}$ is a feasible parent of $u$.
\end{corollary}

\begin{lemma}
\label{lem:spf:propagation:case_subset}
    Let $\Delta_u \subseteq X$.
    Then, $\neighborz{u}$ is a feasible parent of $u$ if $\dist(S,\projz{u}) \leq \dist(S,\projy{u})$, and $\neighbory{u}$ is a feasible parent of $u$ if $\dist(S,\projz{u}) \geq \dist(S,\projy{u})$.
\end{lemma}

\begin{proof}
    Note that $u$ is visible by both $\projz{u}$ and $\projy{u}$ since $\Delta_u \subseteq X$.
    By \Cref{cor:spf:triangle:neighbors}, $\dist(S,u) = \min \{ \dist(S,\neighborz{u}) + 1, \dist(S,\neighbory{u}) + 1 \}$.
    By \Cref{cor:spf:triangle:portal}, $\dist(S,v) = \min_{w \in P_v} \{ \dist(S,w) + \dist(w,v) \}$ for all $v \in \Delta_u$.
    Note that $\dist(w,v)$ is equal for each $v \in \{ \neighborz{u}, \neighbory{u} \}$ and $w \in P_u$.
    If $\min_{w \in P_u} \{ \dist(S,w) \} = \dist(S,\projz{u}) \leq \dist(S,\projy{u})$, then $\dist(S,u) = \dist(S,\neighborz{u}) + 1 \leq \dist(S,\neighbory{u}) + 1$ holds.
    This implies that $\neighborz{u}$ is a feasible parent of $u$.
    If $\min_{w \in P_u} \{ \dist(S,w) \} = \dist(S,\projy{u}) \leq \dist(S,\projz{u})$, then $\dist(S,u) = \dist(S,\neighbory{u}) + 1 \leq \dist(S,\neighbory{u}) + 1$ holds.
    This implies that $\neighbory{u}$ is a feasible parent of $u$.
    Otherwise, $\argmin_{w \in P_u} \{ \dist(S,w) \} \in P_{\neighborz{u}} \cap P_{\neighbory{u}}$ such that $\dist(S,u) = \dist(S,\neighborz{u}) + 1 = \dist(S,\neighbory{u}) + 1$ holds.
    This implies that both $\neighborz{u}$ and $\neighbory{u}$ are feasible parents of $u$.
    In this case, the statement holds trivially.
\end{proof}

\begin{lemma}
\label{lem:spf:propagation:case_not_subset}
    Let $\Delta_u \not\subseteq X$.
    Then, $\neighborz{u}$ is a feasible parent of $u$ if $u$ is visible by $\projz{u}$, and $\neighbory{u}$ is a feasible parent of $u$ if $u$ is visible by $\projy{u}$.
\end{lemma}

\begin{proof}
    Note that $u$ is visible by exactly one of $\projz{u}$ and $\projy{u}$ since $\Delta_u \not\subseteq X$.
    W.l.o.g., let $u$ be visible by $\projz{u}$.
    Suppose the contrary.
    Due to \Cref{cor:spf:triangle:neighbors}, there must be a shortest path through $\neighbory{u}$ (see \Cref{fig:spf:triangle:shortest_path:ny}).
    Consider the path from $\neighborz{u}$ parallel to the axis through $u$ and $\neighbory{u}$ to the first amoebot $v$ on the shortest path.
    This amoebot must exist since $u$ is not visible by $\projy{u}$.
    We exchange the path from $v$ to $u$ through $\neighbory{u}$ with the one through $\neighborz{u}$ without increasing the length of the shortest path (see \Cref{fig:spf:triangle:shortest_path:nz}).
    This is a contradiction to the assumption.
\end{proof}

\begin{figure}[tbp]
    \begin{minipage}[t]{.49\linewidth}
        \centering
        \includegraphics{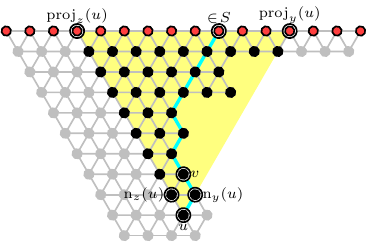}
        \subcaption{Shortest path through $\neighbory{u}$.}
        \label{fig:spf:triangle:shortest_path:ny}
    \end{minipage}
    \hfill
    \begin{minipage}[t]{.49\linewidth}
        \centering
        \includegraphics{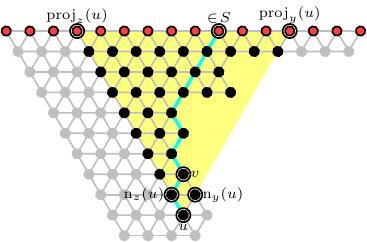}
        \subcaption{Shortest path through $\neighborz{u}$.}
        \label{fig:spf:triangle:shortest_path:nz}
    \end{minipage}
    \caption{
        Rerouting a shortest path.
        The labels refer to the encircled amoebots.
    }
    \label{fig:spf:triangle:shortest_path}
\end{figure}

}{


Consider the \textbf{first phase}, i.e., the propagation into $B' = B \cap \vis(P)$.
Let $\projy{u},\projz{u} \in V_\Delta$ be the projections of $u \in B'$ along the $y$- and $z$-axis onto the axis of $P$, respectively (see \Cref{fig:spf:triangle}).
Let $\Delta_u$ be the set of nodes in $V_\Delta$ within the triangle defined by $u$, $\projy{u}$, and $\projz{u}$.
Let $\neighbory{u},\neighborz{u} \in V_\Delta$ denote $u$'s neighbor into $\projy{u}$'s and $\projz{u}$'s direction, respectively.
Since $u \in \vis(P)$, at least one of $\neighbory{u}$ and $\neighborz{u}$ has to be in $X$.

\begin{figure}[tbp]
    \begin{minipage}[t]{.49\linewidth}
        \centering
        \includegraphics{fig/fig_triangle_01.pdf}
        \subcaption{Case $\Delta_u \subseteq R$.}
        \label{fig:spf:triangle:case_subset}
    \end{minipage}
    \hfill
    \begin{minipage}[t]{.49\linewidth}
        \centering
        \includegraphics{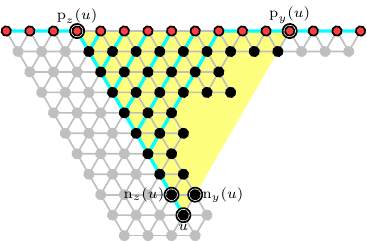}
        \subcaption{Case $\Delta_u \not\subseteq R$.}
        \label{fig:spf:triangle:case_not_subset}
    \end{minipage}
    \caption{
        $(\{u\},P)$-shortest path forest for $P \cup \Delta_u$.
        The labels refer to the encircled amoebots.
    }
    \label{fig:spf:triangle}
\end{figure}

\begin{lemma}
\label{lem:spf:triangle}
    For each $u \in B'$ and $s \in S$, there is a shortest path from $s$ to $u$ within $A \cup P \cup \Delta_u$ (see \Cref{fig:spf:triangle}).
\end{lemma}

\begin{lemma}
\label{lem:spf:propagation:case_subset}
    Let $\Delta_u \subseteq X$.
    Then, $\neighborz{u}$ is a feasible parent of $u$ if $\dist(S,\projz{u}) \leq \dist(S,\projy{u})$, and $\neighbory{u}$ is a feasible parent of $u$ if $\dist(S,\projz{u}) \geq \dist(S,\projy{u})$.
\end{lemma}

\begin{lemma}
\label{lem:spf:propagation:case_not_subset}
    Let $\Delta_u \not\subseteq X$.
    Then, $\neighborz{u}$ is a feasible parent of $u$ if $u$ is visible by $\projz{u}$, and $\neighbory{u}$ is a feasible parent of $u$ if $u$ is visible by $\projy{u}$.
\end{lemma}

}


We compute the parent of each amoebot $u \in B'$ as follows.
First, we have to compute $B'$.
For that, the amoebots establish a circuit for each $y$- and $z$-portal in $P \cup B$ (see \Cref{fig:spf:propagation:circuits:1}).
Each amoebot $u \in P$ beeps on the circuit of $\portal_y(u)$ and $\portal_z(u)$.
Each amoebot $v \in B$ that receives a beep on the circuit of $\portal_y(v)$ ($\portal_z(v)$) is visible by $\projy{v}$ ($\projy{v}$).
Hence, each amoebot in $B$ that does not receive a beep is not in $\vis(P)$ and with that not in $B'$.

Now, each amoebot $u \in B'$ that received a beep on the circuit of $\portal_y(u)$ ($\portal_z(u)$) but not on the circuit of $\portal_z(u)$ ($\portal_y(u)$) chooses $\neighbory{u}$ ($\neighborz{u}$) as its parent (see \Cref{lem:spf:propagation:case_not_subset}).
Then, we apply the PASC algorithm on the shortest path trees in $A \cup P$ to compute $\dist(S,u)$ for each $u \in P$ (see \Cref{fig:spf:propagation:circuits:2}).
Concurrently, each $u \in P$ forwards its distance $\dist(S,u)$ on $\portal_y(u)$ and $\portal_z(u)$ to all amoebots $v \in B'$ that received a beep on both circuits.
This allows each $v \in B'$ that received a beep on both circuits to compare $\dist(S,\projz{u})$ and $\dist(S,\projy{u})$.
It chooses $\neighbory{u}$ as its parent if $\dist(S,\projy{u}) \leq \dist(S,\projz{u})$ and $\neighborz{u}$ otherwise (see \Cref{lem:spf:propagation:case_subset}).

\begin{figure}[tbp]
    \centering
    \includegraphics{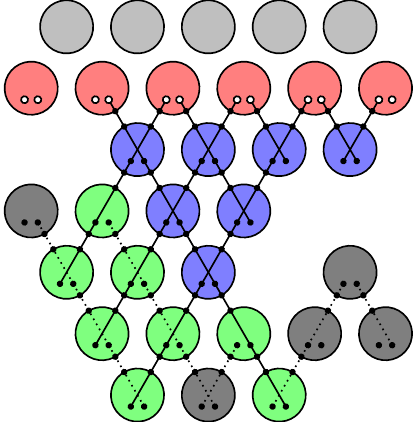}
    \caption{
        Circuits utilized by the propagation algorithm.
        The red amoebots indicate $P$.
        These beep on the white partition sets.
        There is a beep on the- solid circuits while there is no beep on the dotted circuits.
        Each blue amoebot receives two beeps.
        Each green amoebot receives a single beep.
    }
    \label{fig:spf:propagation:circuits:1}
\end{figure}

\begin{figure}[tbp]
    \centering
    \includegraphics{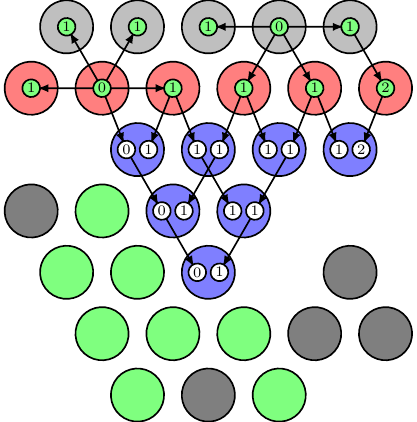}
    \caption{
        Distance computation by the propagation algorithm.
        The amoebots in $A \cup P$ apply the PASC algorithm (green instances), and the amoebots in $P$ forward the results to the amoebots in $B'$ that have received two beeps (white instances).
    }
    \label{fig:spf:propagation:circuits:2}
\end{figure}


Consider the \textbf{second phase}, i.e., the propagation into $B'' = B \setminus \vis(P)$.
$B''$ might consist of several connected components.
We will propagate the shortest path forest into each of those independently of each other.
Let $Z$ denote one of the connected components (see \Cref{fig:spf:propagation:second_phase}).
Let $Z_{B'}$ ($B'_Z$) be the set of all amoebots in $Z$ ($B'$) that are adjacent to an amoebot in $B'$ ($Z$).
Let $s_Z$ be the northernmost amoebot in $Z_{B'}$.

\begin{lemma}
\label{lem:spf:propagation:second_phase:component}
    For each $s \in S$ and $u \in Z$, there is a shortest path from $s$ to $u$ through $s_Z$.
\end{lemma}

\begin{lemma}
\label{lem:spf:propagation:second_phase:source}
    The northermost amoebots in $B'_Z$ are feasible parents of $s_Z$.
\end{lemma}

\iftoggle{full}{
\begin{proof}[Proof of \Cref{lem:spf:propagation:second_phase:component,lem:spf:propagation:second_phase:source}]
    Observe that $B'_Z$ is either a subset of a $y$-portal in $\vis(P)$, a subset of a $z$-portal in $\vis(P)$, or a subset of the union of a $y$- and $z$-portal in $\vis(P)$ (see \Cref{fig:spf:propagation:second_phase}).
    If $B'_Z$ is a subset of a single portal, then let $b_Z$ denote the northernmost amoebot of $B'_Z$.
    If $B'_Z$ is a subset of the union of two portals, then let $b_Z$ denote the intersection between those portals.
    Note that in the latter case, $b_Z \not\in B'_Z$.
    
    For each $u \in B'_Z$, $\Delta_u \not\subseteq X$ holds since otherwise, all its neighbors would be in $P \cup B'$.
    Hence, by \Cref{lem:spf:propagation:case_not_subset}, each amoebot $u \in B'_Z$ chooses its northern neighbor in $B'_Z \cup \{ b_Z \}$ as its parent.
    As a result, the subtree of $b_Z$ in the shortest path forest contains all amoebots in $B'_Z$ (see \Cref{fig:spf:propagation:second_phase}).
    Since \Cref{lem:spf:propagation:case_not_subset} does not depend on $S$, there must be a shortest path from any $v \in A \cup P$ (and with that from any $s \in S \subseteq A \cup P$) to any $w \in B'_Z$ through $b_Z$.
    
    Since each path from $s$ to $u$ must first traverse $B'_Z$ and then $Z_{B'}$, there is a shortest path from $s$ to $u$ through $b_Z$.
    Let $v$ be the first amoebot in $Z_{B'}$ that the shortest path traverses.
    Observe that we can exchange the path from $b_Z$ to $v$ by the shortest path from $b_Z$ to $v$ through $s_Z$ without increasing the length of the shortest path (see \Cref{fig:spf:propagation:second_phase:hourglass}).
    Finally, note that the shortest path from $b_Z$ to $v$ through $s_Z$ can traverse any of the northernmost amoebots in $B'_Z$.
\end{proof}
}{}

By \Cref{lem:spf:propagation:second_phase:source}, $s_Z$ can choose one of the northernmost amoebots in $B'_Z$ as its parent.
\Cref{lem:spf:propagation:second_phase:component} allows us to apply our shortest path tree algorithm on $Z$ with $s_Z$ as the source to compute a parent for each amoebot $u \in Z \setminus \{ s_Z \}$.


\begin{figure}[tbp]
        \centering
        \includegraphics{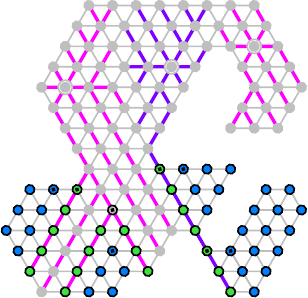}
    \caption{
        Second phase of the propagation algorithm.
        The pink and purple edges indicate the shortest path trees (after the first phase).
        The blue amoebots indicate $B''$.
        For each connected component $Z$ of blue amoebots, the adjacent green amoebots indicate $B'_Z$, the blue amoebot with a dot indicates $s_Z$, and the green or gray amoebot with a dot indicates $b_Z$.
    }
    \label{fig:spf:propagation:second_phase}
\end{figure}

\begin{figure}[tbp]
    \begin{minipage}[t]{.3\linewidth}
        \centering
        \includegraphics{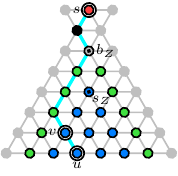}
    \end{minipage}
    \hfill
    \begin{minipage}[t]{.3\linewidth}
        \centering
        \includegraphics{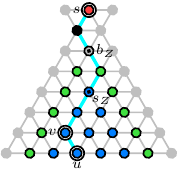}
    \end{minipage}
    \hfill
    \begin{minipage}[t]{.3\linewidth}
        \centering
        \includegraphics{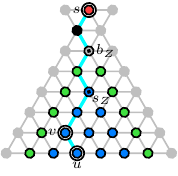}
    \end{minipage}
    \caption{
        Shortest path through $s_Z$.
        The labels refer to the marked amoebots.
    }
    \label{fig:spf:propagation:second_phase:hourglass}
\end{figure}

\begin{lemma}
\label{lem:spf:propagation:algorithm}
    Let $P$ be a portal that divides the amoebot structure into two sides $A$ and $B$. 
    Let $S \subseteq A \cup P$.
    Let a $S$-shortest path forest for $A \cup P$ be given.
    The propagation algorithm computes an $S$-shortest path forest for the whole amoebot structure within $O(\log n)$ rounds.
\end{lemma}

\iftoggle{full}{
\begin{proof}
    The correctness of the first phase follows from \Cref{lem:spf:propagation:case_subset,lem:spf:propagation:case_not_subset}.
    The phase consists of a single round to compute $B'$ and $O(\log n)$ rounds for the PASC algorithm.
    The correctness of the second phase follows from \Cref{lem:spf:propagation:second_phase:source,lem:spf:propagation:second_phase:component,th:spt}.
    By \Cref{th:spt}, the phase requires $O(\log n)$ rounds.
    Overall, the algorithm requires $O(\log n)$ rounds.
\end{proof}
}{}


\subsection{Divide and Conquer Approach}



In this section, we describe divide and conquer approach.
We start with an overview of our shortest path forest algorithm.
First, we split the amoebot structure at portals into smaller regions until each region intersects at most two portals with a source.
Then, we compute a shortest path forest for each of these regions.
Finally, we iteratively merge the regions.
In order to determine disjoint pairs of adjacent regions to merge, we make use of a centroid decomposition tree.
This also minimizes the number of necessary iterations.

In the subsequent subsections, we will elaborate on how to split the amoebot structure into smaller regions, on how to compute a shortest path forest for these regions, and on how to merge the shortest path forests of adjacent regions.


\subsubsection{Dividing the Amoebot Structure}

In this subsection, we consider the splitting of the amoebot structure into smaller regions.
Let $\mathcal Q = \{ P \in \mathcal P \mid P \cap S \neq \emptyset \}$ be the set of all portals that contain at least one source.
In order to compute $\mathcal Q$, each portal establishes a circuit that connects all amoebots of the portal.
Then, each amoebot $u \in S$ beeps on the circuit of $\portal(u)$.
Each amoebot that receives a beep belongs to a portal with at least one source.
Next, we apply the \iftoggle{full}{root and prune primitive}{augmentation primitive} to compute $A_{\mathcal Q}$ (see \Cref{lem:portal:augmentation}).
Let $\mathcal Q' = \mathcal Q \cup A_{\mathcal Q}$.

\begin{lemma}
\label{lem:spf:dividing_step:1}
    Our shortest path forest algorithm computes set $\mathcal Q'$ within $O(\log k)$ rounds.
\end{lemma}

\iftoggle{full}{
\begin{proof}
    The computation of $\mathcal Q$ only requires a single round.
    Observe that $|\mathcal Q| = O(k)$.
    The correctness and runtime follows from \Cref{lem:portal:augmentation}.
\end{proof}
}{}

\iftoggle{full}{


Each of these portals splits the amoebot structure into two regions, with the portal being part of both regions (see \Cref{fig:spf:splitting:initial,fig:spf:splitting:portals}).
In general, the resulting regions may intersect an arbitrary number of portals in $\mathcal Q'$.
We split the regions further until each region intersects one or two portals in $\mathcal Q'$ as follows.

\begin{figure}[tbp]
    \begin{minipage}[t]{\linewidth}
        \centering
        \includegraphics{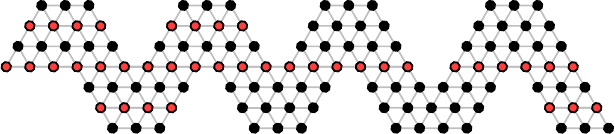}
        \subcaption{Initial amoebot structure.}
        \label{fig:spf:splitting:initial}
    \end{minipage}
    
    \bigskip

    \begin{minipage}[t]{\linewidth}
        \centering
        \includegraphics{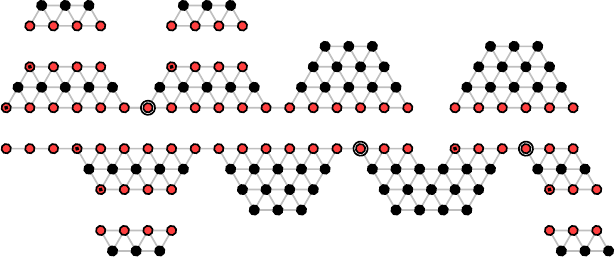}
        \subcaption{Regions after splitting the amoebot structure at each portal in $\mathcal Q$.}
        \label{fig:spf:splitting:portals}
    \end{minipage}
    
    \bigskip

    \begin{minipage}[t]{\linewidth}
        \centering
        \includegraphics{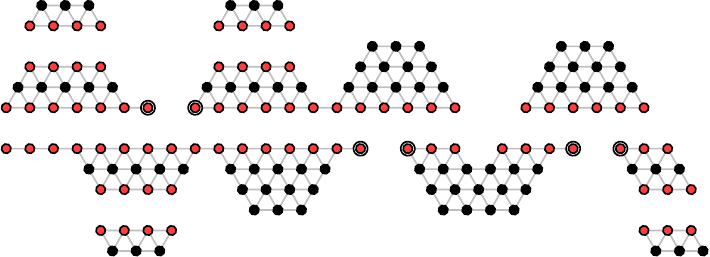}
        \subcaption{Regions after splitting the regions at each marked amoebot.}
        \label{fig:spf:splitting:faces}
    \end{minipage}
    \caption{
        Regions.
        The red amoebots indicate the portals in $\mathcal Q'$.
        The encircled amoebots indicate the marked amoebots.
        The amoebots with a dot were marked but then unmarked.
    }
    \label{fig:spf:splitting}
\end{figure}

In each region, each portal $P_1 \in \mathcal Q'$ marks amoebot $\connector_{P_1}(P_2)$ for each $P_2 \in V_{\mathcal Q}$.
Then, each portal unmarks the westernmost marked amoebot.
For that, each portal establishes a circuit between its endpoints that it cuts at each marked amoebot.
The westernmost amoebot sends a beep which is received by the westernmost marked amoebot.
Finally, we split each region at each still marked amoebot into two regions, with the marked amoebot being part of both regions (see \Cref{fig:spf:splitting:faces}).

}{


Each of these portals splits the amoebot structure into two regions, with the portal being part of both regions (see \Cref{fig:spf:splitting:initial,fig:spf:splitting:portals}).
The resulting regions may intersect an arbitrary number of portals in $\mathcal Q'$.
We split the regions further until each region intersects one or two portals in $\mathcal Q'$.
For that, each portal in $\mathcal Q'$ marks a minimal number of amoebot in its ``bottlenecks'' (see \Cref{fig:spf:splitting:portals}).
Then, we split each region at each marked amoebot into two regions, with the marked amoebot being part of both regions (see \Cref{fig:spf:splitting:faces}).

\begin{figure}[tbp]
    \begin{minipage}[t]{\linewidth}
        \centering
        \includegraphics{fig/fig_regions_01.pdf}
        \subcaption{Initial amoebot structure.}
        \label{fig:spf:splitting:initial}
    \end{minipage}
    
    \bigskip

    \begin{minipage}[t]{\linewidth}
        \centering
        \includegraphics{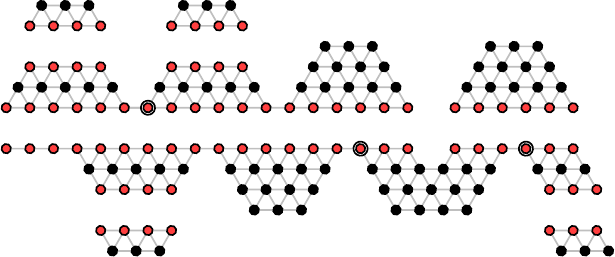}
        \subcaption{Regions after splitting the amoebot structure at each portal in $\mathcal Q$.}
        \label{fig:spf:splitting:portals}
    \end{minipage}
    
    \bigskip

    \begin{minipage}[t]{\linewidth}
        \centering
        \includegraphics{fig/fig_regions_03.pdf}
        \subcaption{Regions after splitting the regions at each marked amoebot.}
        \label{fig:spf:splitting:faces}
    \end{minipage}
    \caption{
        Regions.
        The red amoebots indicate the portals in $\mathcal Q'$.
        The encircled amoebots indicate the marked amoebots in the ``bottlenecks''.
    }
    \label{fig:spf:splitting}
\end{figure}

}

\begin{lemma}
\label{lem:spf:dividing_step:2}
    Our shortest path forest algorithm decomposes the amoebot structure into regions that intersect one or two portals in $\mathcal Q'$ within $O(1)$ rounds.
\end{lemma}

\iftoggle{full}{
\begin{proof}
    Note that each region intersects at least one portal in $\mathcal Q'$ since whenever we split a region at a portal in $\mathcal Q'$ or at an amoebot of a portal in $\mathcal Q'$, it becomes part of both regions.

    Consider the portal graph with root $R$ and without any subtrees with portals in $\mathcal Q$.
    We first split the amoebot structure at each portal in $\mathcal Q'$.
    This is equivalent to replacing each portal $P \in \mathcal Q'$ with two portals $P_N$ and $P_S$ where we assign all incident edges of $P$ to the north to $P_N$, and all incident edges of $P$ to the south to $P_S$.
    We will consider both $P_N$ and $P_S$ to be in $\mathcal Q'$.
    
    Then, we split each region at the marked amoebots.
    This is equivalent to replacing each portal $P \in \mathcal Q'$ with $\degree_{\mathcal Q}(P)$ subportals where we assign each incident edge of $P$ to one of the subportals.
    Again, we will consider each of the $\degree_{\mathcal Q}(P)$ subportals to be in $\mathcal Q'$.
    In the resulting portal graph, each (sub)portal $P \in \mathcal Q'$ has $\degree_{\mathcal Q}(P) = 1$.

    Further, we claim that each connected component has at most two (sub)portals in $\mathcal Q'$.
    Suppose the contrary, i.e., there is a connected component with more than two (sub)portals in $\mathcal Q'$.
    Since each (sub)portal $P \in \mathcal Q'$ has $\degree_{\mathcal Q}(P) = 1$, it must be a leaf of the connected component.
    Since by the assumption, there are at least $3$ leaves, there must be a portal $P$ with $\degree_{\mathcal Q}(P) \geq 3$.
    This is a contradiction since $P \in \mathcal Q'$ and $P$ is not a leaf.
    The correctness of the lemma follows from the claim since the connected components are the portal graphs of the regions.

    Finally, note that we have already performed the root and prune primitive.
    Hence, we have all necessary information to decompose the amoebot structure.
    We only need a single round to unmark the westernmost marked amoebot of each portal.
\end{proof}
}{}

\subsubsection{Base Case}

In this subsection, we explain how to compute a shortest path forest for our base case, i.e., regions whose boundary intersects one or two portals in $\mathcal Q'$.
Our goal is to compute an $(S \cap \region)$-shortest path forest for each region $\region$.

In a preprocessing step, we determine for each region whether it intersects one or two portals in $\mathcal Q'$.
First, we apply the election primitive to elect a portal $R' \in \mathcal Q'$ (see \Cref{lem:portal:election}).
Then, we apply the root and prune primitive to root the portal tree at $R'$ (and prune any subtree without a portal in $\mathcal Q'$) (see \Cref{lem:portal:rap}).
Let $\portal(\region) = \{ P \in V_{\mathcal P} \mid P \cap \region \neq \emptyset \}$ denote the set of all portals that intersect region $\region$.
Let $\mathcal Q'_\region = \mathcal Q' \cap \portal(\region)$ denote the set of all portals in $\mathcal Q'$ that intersect $\region$.

\begin{lemma}
\label{lem:spf:ancestor}
    For each region $\region$, the lowest common ancestor of $\portal(\region)$ with respect to $R'$ is in $\mathcal Q'_\region$.
\end{lemma}

\iftoggle{full}{
\begin{proof}
    The statement holds trivially if $R' \in \mathcal Q'_\region$.
    Otherwise, $R' \not\in \portal(\region)$ holds.
    By definition of the regions, the portals in $\portal(\region)$ can only be reached from $R'$ through one of the portals in $\mathcal Q'_F$ which must be the lowest common ancestor of $\portal(\region)$.
\end{proof}
}{}

A portal in $\mathcal Q'_\region$ identifies as the lowest common ancestor if it is $R'$ or its parent is not in $\portal(\region)$.
Let $P^\mathit{LCA}_\region \in Q'_\region$ denote that lowest common ancestor portal of $\portal(\region)$.
Let $P^\mathit{DSC}_\region \in Q'_\region$ denote the other portal (descendant) if it exists.
Each region $\region$ establishes a circuit that connects all amoebots of the region.
If $P^\mathit{DSC}_\region$ exists, it beeps on the circuit.
Clearly, the region intersects two portals in $\mathcal Q'$ if there is a beep, and only one portal in $\mathcal Q'$ otherwise.

Now, we compute the shortest path forests for each region $Y$ as follows.
If $\region$ intersects only one portal $P^\mathit{LCA}_\region \in \mathcal Q'$, we proceed as follows.
First, we apply the line algorithm on $P^\mathit{LCA}_\region \cap \region$ to compute an $(S \cap P^\mathit{LCA}_\region)$-shortest path forest for $P^\mathit{LCA}_\region$.
Then, we apply the propagation algorithm to compute an $(S \cap P^\mathit{LCA}_\region)$-shortest path forest for $\region$.
Note that $S \cap P^\mathit{LCA}_\region = S \cap \region$.

If $\region$ intersects two portals $P^\mathit{LCA}_\region, P^\mathit{DSC}_\region \in \mathcal Q'$, we proceed as follows.
First, we apply the previous procedure on $P^\mathit{LCA}_\region$ to obtain an $(S \cap P^\mathit{LCA}_\region)$-shortest path forest for $\region$.
Then, we repeat the previous procedure on $P^\mathit{DSC}_\region$ to obtain an $S \cap P^\mathit{DSC}_\region$-shortest path forest for $\region$.
Finally, we apply the merging algorithm to compute an $(S \cap (P^\mathit{LCA}_\region \cup P^\mathit{DSC}_\region))$-shortest path forest for $\region$.
Note that $S \cap (P^\mathit{LCA}_\region \cup P^\mathit{DSC}_\region) = S \cap \region$.

\begin{lemma}
\label{lem:spf:base_case}
    Our shortest path forest algorithm computes an $(S \cap \region)$-shortest path forest for each $\region$ within $O(\log n)$ rounds.
\end{lemma}

\iftoggle{full}{
\begin{proof}
    The correctness of the preprocessing step follows from \Cref{lem:portal:election,lem:portal:rap,lem:spf:ancestor}.
    The correctness of the remaining procedure follows from \Cref{lem:spf:line,lem:spf:propagation:algorithm,lem:spf:merging:algorithm}.
    By \Cref{lem:portal:election,lem:portal:rap}, the preprocessing step requires $O(\log k)$ rounds.
    By \Cref{lem:spf:line,lem:spf:propagation:algorithm,lem:spf:merging:algorithm}, the remaining procedure requires $O(\log n)$ rounds.
    Altogether, $O(\log n)$ rounds are necessary.
\end{proof}
}{}


\subsubsection{Merging Step}

In this subsection, we explain the merging step of our shortest path forest algorithm.
Let $P$ be a portal.
Let $\mathcal\region_P$ denote the set of all regions that intersect $P$.
Let $\region' = \bigcup_{\region \in \mathcal\region_P} \region$.
We merge the $(S \cup \region)$-shortest path forest of each $\region \in \mathcal\region_P$ into an $(S \cup \region')$-shortest path forest for $\region'$ in two phases as follows.

In the first phase, we iteratively merge all regions on both sides of $P$, respectively, as follows.
Consider the regions on one side.
Recall that by construction, these regions are separated by marked amoebots in $P$.
Initially, let $M$ denote the set of all marked amoebots.
We will remove amoebots from $M$ after each iteration.
Each iteration consists of three steps.

In the first step, we check the termination condition, i.e., whether we have already merged all regions, as follows.
For that, we have to check whether there are any marked amoebots left.
First, the amoebots establish a circuit that connects all amoebots in $P$.
Then, each marked amoebot beeps.
Clearly, we terminate the first phase if there is no beep.
Otherwise, we proceed to the next step.

\iftoggle{full}{

In the second step, we divide the regions into pairs as follows.
We apply a single iteration of the PASC algorithm on $P$ with $M$ (see \Cref{fig:spf:pairing}).
Each amoebot in $M$ obtains the parity of its prefix sum.
Let $M'$ denote the amoebots that have an odd parity.
For each amoebot in $M'$, we pair the regions that are separated by it.

\begin{figure}[tbp]
    \centering
    \includegraphics{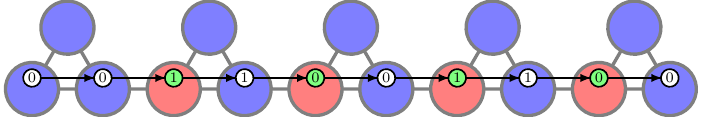}
    \caption{
        Pairing of regions.
        The red amoebots indicate $M$.
    }
    \label{fig:spf:pairing}
\end{figure}

In the third step, we merge each pair as follows.
Consider an amoebot $s \in M'$.
Let $\region_W$ and $\region_E$ denote the region to the west and east of $s$, respectively.
Observe that each shortest path from $S \cap \region_W$ to $\region_E$ and each shortest path from $S \cap \region_E$ to $\region_W$ has to traverse $s$.
Similarly to the second phase of the propagation algorithm, this allows us to apply our shortest path tree algorithm with $s$ as the source to propagate the shortest path forest of both regions into the other region, respectively.
Then, we apply the merging algorithm to compute an $(S \cap \region'')$-shortest path forest for $\region'' = \region_W \cup \region_E$.
Finally, we remove $M'$ from $M$ and proceed with the next iteration.

}{

In the second step, we divide the regions into pairs as follows.
We apply the PASC algorithm to compute a set $M'$ that contains every second amoebot in $M$.
For each amoebot in $M'$, we pair the regions that are separated by it.

In the third step, we merge each pair as follows.
Consider an amoebot $s \in M'$.
Let $\region_W$ and $\region_E$ denote the region to the west and east of $s$, respectively.
Observe that each shortest path from $S \cap \region_W$ to $\region_E$ and each shortest path from $S \cap \region_E$ to $\region_W$ has to traverse $s$.
This allows us to apply our shortest path tree algorithm with $s$ as the source to propagate the shortest path forest of both regions into the other region, respectively.
Then, we apply the merging algorithm to compute an $(S \cap \region'')$-shortest path forest for $\region'' = \region_W \cup \region_E$.
Finally, we remove $M'$ from $M$ and proceed with the next iteration.

}

After the first phase has terminated, we are left with one region to the north and south of $P$, respectively.
Let $\region_N$ and $\region_S$ denote these regions.
In the second phase, we merge both regions as follows.
First, we apply the propagation algorithm to compute an $(S \cap \region_N)$-shortest path forest for $\region''' = \region_N \cup \region_S$.
Then, we apply the propagation algorithm again to compute an $(S \cap \region_S)$-shortest path forest for $\region'''$.
Finally, we apply the merging algorithm to compute an $(S \cap \region''')$-shortest path forest for $\region'''$.

\begin{lemma}
\label{lem:spf:merging_step}
    Let $P$ be a portal.
    Let $\mathcal\region_P$ denote the set of all regions that intersect $P$.
    Let $\region' = \bigcup_{\region \in \mathcal\region_P} \region$.
    Our shortest path forest algorithm merges the $(S \cup \region)$-shortest path forest of each $\region \in \mathcal\region_P$ into an $(S \cup \region')$-shortest path forest within $O(\log n \log k)$ rounds.
\end{lemma}

\iftoggle{full}{
\begin{proof}
    The correctness follows from \Cref{th:spt}, \Cref{lem:spf:propagation:algorithm,lem:spf:merging:algorithm}.
    Note that $|\mathcal\region_P| = O(k)$ since $\mathcal Q' = O(k)$.
    Since each iteration of the first phase halves the number of regions, we need $O(\log k)$ iterations.
    The first and second step of each iteration only require a constant number of rounds.
    By \Cref{th:spt} and \Cref{lem:spf:merging:algorithm}, the third step of each iteration requires $O(\log n)$ rounds.
    Hence, the first phase requires $O(\log n \log k)$.
    By \Cref{lem:spf:propagation:algorithm,lem:spf:merging:algorithm}, the second phase requires $O(\log n)$ rounds.
\end{proof}
}{}


\subsubsection{Putting everything together}

\iftoggle{full}{


We are now (almost) able to put our shortest path forest algorithm together.
First, we compute a set $\mathcal Q'$ of portals (see \Cref{lem:spf:dividing_step:1}).
Then, we decompose the amoebot structure along the portals in $\mathcal Q'$ into smaller regions (see \Cref{lem:spf:dividing_step:2}), and compute a shortest path forest for each of these regions (see \Cref{lem:spf:base_case}).
Finally, we iteratively merge the regions along subsets of $\mathcal Q'$ (see \Cref{lem:spf:merging_step}).

So, it only remains to define and compute the subsets for each iteration.
The only requirement for the subsets is that each region intersects at most one portal of the subset.
Otherwise, we would have to merge a region along two portals at once.

In order to minimize the number of iterations, we utilize a $\mathcal Q'$-centroid decomposition tree.
Note that by definition, the $\mathcal Q'$-centroids of the same depth are separated by the $\mathcal Q'$-centroids of the previous depth.
This satisfies our requirement.

Unfortunately, since we have to merge the regions from the leaves of the $\mathcal Q'$-centroid decomposition tree to its root, and since we are not able to store the depths of the $\mathcal Q'$-centroids, we have to recompute the $\mathcal Q'$-centroid decomposition tree for each iteration (see \Cref{lem:portal:decomposition}).
Note that we always compute the same $\mathcal Q'$-centroid decomposition tree since our decomposition primitive is deterministic.
In order to identify the correct $\mathcal Q'$-centroids for the current iteration, we utilize the binary counter technique by Padalkin et al.\ \cite{DBLP:conf/dna/PadalkinSW22}.
We refer to \cite{DBLP:conf/dna/PadalkinSW22} for more details.

\begin{theorem}
    The shortest path forest algorithm computes an $S$-shortest path forest within $O(\log n \log^2 k)$ rounds.
\end{theorem}

\iftoggle{full}{
\begin{proof}
    The correctness follows from \Cref{lem:spf:dividing_step:1,lem:spf:dividing_step:2,lem:spf:base_case,lem:spf:merging_step,lem:portal:decomposition}.
    By \Cref{lem:spf:dividing_step:1}, computing $\mathcal Q'$ requires $O(\log k)$ rounds.
    By \Cref{lem:spf:dividing_step:2}, decomposing the amoebot structure into smaller regions requires $O(1)$ rounds.
    By \Cref{lem:spf:base_case}, computing a shortest path forest for each of these regions requires $O(\log n)$ rounds.
    Each iteration of the merging phase consists of computing a set of portals and of merging regions along these portals.
    By \Cref{lem:portal:decomposition}, the former requires $O(\log^2 k)$ rounds, and by \Cref{lem:spf:merging_step}, the latter $O(\log n \log k)$ rounds.
    By \Cref{lem:decomposition:height}, we need $O(\log k)$ iteration.
    Overall, the shortest path algorithm requires $O(\log n \log^2 k)$ rounds.
\end{proof}
}{}

So far, the shortest path forest algorithm has ignored the set $D$ of destinations.
In order to obtain an $(S,D)$-shortest path forest, it applies the root and prune primitive on each shortest path tree $T_s$ with $s$ and $D$ as parameters (compare to \Cref{sec:spt}).

\begin{corollary}
    The shortest path forest algorithm computes an $(S,D)$-shortest path forest within $O(\log n \log^2 k)$ rounds.
\end{corollary}

}{


We are now able to put our shortest path forest algorithm together.
First, we compute a set $\mathcal Q'$ of portals (see \Cref{lem:spf:dividing_step:1}).
Then, we decompose the amoebot structure along the portals in $\mathcal Q'$ into smaller regions (see \Cref{lem:spf:dividing_step:2}), and compute a shortest path forest for each of these regions (see \Cref{lem:spf:base_case}).
Finally, we iteratively merge the regions along subsets of $\mathcal Q'$ (see \Cref{lem:spf:merging_step}).
In order to minimize the number of iterations, we utilize a $\mathcal Q'$-centroid decomposition tree (see \Cref{lem:portal:decomposition}).

So far, the shortest path forest algorithm has ignored the set $D$ of destinations.
In order to obtain an $(S,D)$-shortest path forest, it applies the root and prune primitive on each shortest path tree $T_s$ with $s$ and $D$ as parameters (compare to \Cref{sec:spt}).

\begin{theorem}
    The shortest path forest algorithm computes an $(S,D)$-shortest path forest within $O(\log n \log^2 k)$ rounds.
\end{theorem}

}



\section{Conclusion and Future Work}

In this paper, we have shown how to construct shortest path forests in polylogarithmic time.
However, we are not aware of any non-trivial lower bounds and leave their investigation for future work.
Furthermore, the presented algorithms do not work on amoebot structures with holes since \Cref{lem:portal_graph:tree,lem:portal_graph:triangular} do not hold anymore.
Hence, it would be interesting to consider shortest path problems in general amoebot structures.


\bibliography{literature}

\end{document}